\pdfoutput=1

\documentclass[10pt,journal,compsoc]{IEEEtran}
\usepackage{xr}
\usepackage{graphicx}
\usepackage{bbm}
\usepackage{soul}
\renewcommand\hl[1]{#1} 
\usepackage{amsthm}
\usepackage{chngcntr}
\usepackage{latexsym}
\usepackage{amssymb}
\usepackage{amsmath}
\usepackage{cite}
\usepackage{verbatim}
\usepackage{multirow}
\usepackage{lastpage}
\usepackage{calc}
\usepackage{eso-pic}
\usepackage{tabularx}
\usepackage{ifthen}
\usepackage{epstopdf}
\usepackage[shortlabels]{enumitem}
\usepackage{textcomp}
\usepackage[caption=false]{subfig}
\usepackage{wrapfig}
\usepackage[font=small,belowskip=0pt,aboveskip=0pt]{caption}
\newtheorem{theorem}{Theorem}
\usepackage[english]{babel}
\usepackage{textgreek}
\usepackage[utf8x]{inputenc}
\usepackage[linesnumbered,ruled,commentsnumbered]{algorithm2e}

\newcounter{def}

\let\oldnl\nl
\newcommand{\nonl}{\renewcommand{\nl}{\let\nl\oldnl}}

\newtheorem{definition}[def]{Definition}

\ifCLASSINFOpdf
\else
\fi
\IEEEoverridecommandlockouts
\IEEEpubid{\makebox{1536-1233 \copyright 2021 IEEE. Citation information: DOI 10.1109/TMC.2021.3071468, IEEE
Transactions on Mobile Computing
 \hfill} \hspace{\columnsep}\makebox{ }}

\begin{document}

\title{Enabling Content-Centric Device-to-Device Communication in the Millimeter-Wave Band}

\author{Niloofar~Bahadori, Mahmoud Nabil, Brian Kelley, Abdollah Homaifar

\IEEEcompsocitemizethanks{\IEEEcompsocthanksitem N. Bahadori, M.N. Mahmoud and A. Homaifar are with the Department
of Electrical and Computer Engineering, North Carolina A\&T State University , Greensboro, NC \protect\\
E-mail: nbahador@aggies.ncat.edu,\{mnmahmoud, homaifar\}@ncat.edu
\IEEEcompsocthanksitem B. Kelley is with the Department
of Electrical Engineering, University of Texas at San Antonio, San Antonio, TX \protect\\
E-mail: brian.kelley@utsa.edu
\IEEEcompsocthanksitem The corresponding author is Dr. Abdollah Homaifar.
}
}

\markboth{IEEE TRANSACTIONS ON MOBILE COMPUTING}%
{Bahadori \MakeLowercase{\textit{et al.}}: Enabling Content-Centric Device-to-Device Network in the Millimeter-Wave Band}

\IEEEtitleabstractindextext{%
\begin{abstract}
The growth in wireless traffic and mobility of devices have congested the core network significantly. This bottleneck, along with spectrum scarcity, made the conventional cellular networks insufficient for the dissemination of large contents. The ability of content-centric networking (CCN) and device-to-device (D2D) communication in offloading the network and huge unlicensed spectrum at millimeter-wave (mmWave) band, make the integration of CCN with D2D communication in the mmWave band a viable solution to improve the network’s \hl{throughput}. In this paper, we propose a novel scheme that enables efficient initialization of CCN-based D2D networks in the mmWave band through addressing decentralized D2D peer association and antenna beamwidth selection.  The proposed scheme considers mmWave characteristics such as directional communication and blockage susceptibility. We propose a heuristic peer association algorithm to associate D2D users using context information, including link stability time and content availability. We model the beamwidth selection problem as a potential game and propose a synchronous log-linear learning algorithm to obtain the game’s optimal Nash equilibrium. The performance of the proposed scheme in terms of data throughput and transmission efficiency is evaluated through extensive simulations. Simulation results show that the proposed scheme improves network performance significantly and outperforms other methods in the literature.
\end{abstract}

}
\maketitle
\IEEEdisplaynontitleabstractindextext

\IEEEpeerreviewmaketitle
\IEEEraisesectionheading{\section{Introduction}}

\IEEEPARstart{T}{he} continuous growth in the number of smart mobile devices and multimedia services lead to unprecedented growth of data traffic in wireless networks. The capacity of the wireless network cannot practically cope with the tremendous growth in mobile traffic due to the congested core network and spectrum scarcity in the microwave band \cite{agiwal2016next}.
Studies on cellular network traffic show that a significant portion of the mobile traffic is due to duplicate downloads of a few popular contents (e.g., popular videos) with large sizes \cite{liu2014content}. In light of this paradigm shift, the majority of the next-generation communications will be content-oriented. In other words, the flow of data through the network is driven by the content of the data, rather than by explicit addresses of the hosts of the data.

The shift from the connection-centric network to a more content-centric network (CCN)\cite{liu2014content} along with the advances in device-to-device (D2D) communications\cite{asadi2014survey} motivates caching the popular contents in the edge devices with relatively large storage sizes. 
Integrating CCN and D2D communications thus enabling mobile users to access the popular content over direct links from nearby users rather than the cellular network, is envisioned to improve the spectrum efficiency by offloading the cellular network \cite{liu2017information}. The decentralized architecture of the CCN-based D2D network makes the communication network more robust, flexible, and efficient. \hl{For example, in emergency or military scenarios where the communication infrastructure is unavailable, a CCN-based D2D network can be established to disseminate and share contents among network nodes.}
Despite its potential advantages, the large-scale implementation of the D2D communications in the CCN network has yet to be realized, mainly due to severe multi-user interference (MUI) and lack of enough bandwidth in the microwave band.
Exploiting the millimeter wave (mmWave) band for CCN-based D2D communications is seen as an attractive solution, where directional communication alleviates the problem of MUI, and abundant unlicensed spectrum addresses the spectrum scarcity issue \cite{rappaport2013millimeter}.

However, before reaping the potential advantages of CCN-based D2D communication in the mmWave band, one needs to address several new technical challenges. Initialization is crucial to establish reliable physical D2D links between communicating nodes. Initialization is a sensitive control layer procedure that requires careful planning as it may impose significant delay and overhead to the network, which in turn reduces the network's throughput \cite{giordani2016initial}.
Initialization in the directional CCN-based D2D network requires two main steps, namely, peer association and antenna beam management.
Peer association enables mobile users to discover a corresponding peer that is cached with their desired content, while beam management controls the beam alignment and the width of the antenna beam. Most of the existing mmWave D2D initialization schemes mainly focused on either peer association \cite{xiao2016energy, gu2015matching, zhao2016matching,namvar2015context,zhang2014social,lee2016new,wang2016propagation} or beam management problem \cite{ kutty2015beamforming,bahadori2018device,perfecto2017millimeter,perfecto2016beamwidth}. Even the work that considers both problems simultaneously is based on peer-to-peer protocol rather than CCN protocol \cite{zhang2018power}. Besides, most of the existing work oversimplify the model of either problem or both. In particular, existing context-aware methodologies proposed for peer association in microwave band such as \cite{zhao2016matching,namvar2015context,zhang2014social,lee2016new,wang2016propagation} fail to consider the impact of susceptibility to blockages, directional communication, and mobility of users on peer association efficiency.
Beam management in mmWave directional communication has also been discussed in the literature.
However, most of the existing schemes focus on beam alignment \cite{kutty2015beamforming, bahadori2018device}, while antenna beamwidth selection, despite its significant impact on the network's data throughput, has not been explored vastly in the literature. Moreover, existing beam management solutions rely heavily on a central controller \cite{perfecto2016beamwidth,perfecto2017millimeter,zhang2018power}.

In this paper, we propose a decentralized initialization scheme for enabling the CCN-based D2D network in the mmWave band through addressing peer association and antenna beamwidth selection. \hl{In the CCN-based D2D network, CCN protocol and D2D communication are used to offload the cellular network and high-bandwidth mmWave links are implemented to enhance the network's sum-rate capacity, particularly on the network fronthaul.}
We consider the limitations of mmWave band propagation such as blockage susceptibility, directional communication links and mobility of users. D2D users in our proposed model can be divided into two categories, namely, D2D transmitters (DTs) and D2D requesters (DRs). The proposed scheme enables users to perform peer association by utilizing context information, including data segment availability and link stability. The former parameter determines the amount of the desired data that is cached in the DTs, and the latter captures the time that the directional D2D link is stable for data transmission. Following peer association, D2D users are required to select the proper antenna beamwidth by considering the size of the requested data and the trade-off between data throughput and antenna beamwidth. Tools from game theory are used to model the beamwidth selection problem. The existence of the game's steady-state solution (i.e., Nash equilibrium) is established within a potential game framework. Further, a synchronous Log-linear learning (LLL) based algorithm is proposed to enable users to optimize their antenna beamwidth.
The main contributions of this work are summarized as follows:
\begin{itemize}
    \item A novel decentralized scheme is proposed to enable the initialization process in the CCN-based D2D network at mmWave frequencies, \hl{which considers the mmWave band propagation limitations.} The proposed scheme consists of two phases, namely, heuristic peer association algorithm and synchronous LLL-based beamwidth selection algorithm. Our extensive analysis shows that the proposed scheme improves the network’s throughput significantly.
    \item A heuristic peer association algorithm is proposed to enable DRs to discover their corresponding DT using context information, including data segment availability and stability time of the directional links.
    \hl{Unlike existing peer association methods \mbox{\cite{xiao2016energy, gu2015matching, zhao2016matching,namvar2015context,zhang2014social,lee2016new,wang2016propagation}}, the proposed algorithm considers blockage susceptibility, users' mobility and content availability simultaneously.} Moreover, compared to the existing methods, such as the deferred acceptance (DAA) algorithm \cite{perfecto2017millimeter,namvar2015context}, the proposed algorithm has a low-overhead with a low-computational load.
    \item The problem of D2D pair optimal beamwidth selection is modeled using game-theoretic approaches with a well-defined utility function. Moreover, we prove that the beamwidth selection game is an exact potential game to which the optimal Nash equilibrium is obtained. Further, a synchronous LLL-based algorithm is proposed to obtain the optimal Nash equilibrium of the game. The convergence of the LLL-based beamwidth selection algorithm has significantly accelerated, thanks to the short-range and directional mmWave communication. Compared to the particle swarm optimization (PSO) algorithm \cite{perfecto2017millimeter}, the proposed algorithm is decentralized and is guaranteed to converge to the optimal solution.
\end{itemize}
The remainder of this paper is organized as follows. Section \ref{sec:RelatedWorks}  reviews the relevant related work. The system model and assumptions are described in Section \ref{sec:systemModel}. The network data throughput maximization problem is formulated in Section \ref{sec:problemFormulation}. A novel decentralized scheme for enabling CCN-based D2D scheme in mmWave band through peer association and beamwidth selection is proposed in Section \ref{sec:ProposedScheme}. Simulation results are presented in Section \ref{sec:NumericalResults} and finally, conclusions are drawn in Section \ref{sec:conclusion}.

\section{Related Work}\label{sec:RelatedWorks}
CCN-based D2D communication has attracted a considerable amount of attention to bring the popular contents closer to the end-users, and thereby, to increase network capacity and to improve data throughput. Several schemes have addressed communication issues in CCN-based D2D networks, including content caching, link allocation and forwarding strategies \cite{liu2017information}.
However, despite the significant research on CCN-based D2D communication, very few in the literature have attempted to utilize the huge unlicensed mmWave bandwidth for transmitting large-size data files \cite{giatsoglou2017d2d,niu2015exploiting}. Among the challenges that face directional CCN-based D2D networks is establishing an efficient initialization procedure.
Initialization is crucial in enabling D2D users to establish physical links. Initialization comprises two phases: peer association and beam management, each of which has been explored in the literature separately.

Traditionally a D2D user is matched with a D2D peer in its vicinity, either randomly or based on distance-based algorithms with the goal to maximize the energy efficiency and reduce interference on cellular users \cite{xiao2016energy,  gu2015matching}. Such schemes are not efficient for CCN-based network band as the mobility of users and cached content in D2D devices are ignored. To address these challenges context-aware peer association algorithms are proposed \cite{namvar2015context,wang2016propagation,zhao2016matching}. Authors in \cite{namvar2015context} proposed a context-aware peer association algorithm with the goal of offloading the cellular network that exploits context information about the users' velocity and size of their demanded data to match D2D users using DAA algorithm. A peer propagation- and mobility-aware D2D association algorithm is suggested in \cite{wang2016propagation} based on joint consideration of social graphs, content propagation, and user mobility. \hl{However, most of these algorithms are centralized and none of them considered mmWave characteristics such as directional communication, mobility of users and blockage attenuation in matching D2D users.}

Directional transmissions are used in mmWave band to compensate for the high path-loss \cite{rappaport2013millimeter}.
Therefore, beam management must be implemented in order to establish high-throughput physical links. Beam management comprises two phases, beam alignment and antenna optimal beamwidth selection. Various approaches of beam alignment, which involves aligning the main-lobe of a pair of communicating node, have been suggested in the literature \cite{kutty2015beamforming,bahadori2018device}. Despite the recent advances in beamwidth tuning technologies \cite{wei2019noma} and significant impact of antenna beamwidth on the network performance \cite{perfecto2016beamwidth}, beamwidth selection has yet to be explored properly. The system throughput of a relaying small-cell network is noticeably improved using a coordinated heuristically optimized beamwidth selection \cite{perfecto2017millimeter,perfecto2016beamwidth}.
In addition, the problem of aggravated interference is addressed through proposing device association and beamwidth selection in \cite{zhang2018power}. All of these works have used the particle swarm optimization (PSO) algorithm, which requires a central controller (base station) to optimizes the antenna beamwidth.

Existing research work focus is mainly on either peer association problem or beamwidth selection using centralized approaches, and most are suitable for stationary scenarios. Even the work that considers both issues simultaneously \cite{zhang2018power} is based on peer-to-peer protocol rather than CCN protocol. Therefore, addressing both problems in the CNN framework using a low-overhead decentralized approach is lacking in the literature.
To address these challenges, we propose a novel initialization scheme that aims at maximizing the network sum data throughput, through matching users and selecting the optimal antenna beamwidth for communication. The proposed scheme considers the mmWave band propagation characteristics, such as blockage susceptibility, directional antenna beams, and mobility of users.
\hl{The proposed scheme by considering the content availability and utilizing context information, including users' trajectory and size of the requested data, enables D2D users to find a strategy that maximizes the network's data throughput.}

\begin{table*}[t!]
	\centering
	\begin{minipage}{.95\textwidth}
		{\scriptsize
			\renewcommand{\arraystretch}{1.2}
			\caption{Summary of Notations}\label{tab:Notation}
			\centering
			\begin{tabular}{|p{1.5cm}|p{6cm}|p{1.5cm}|p{6.3cm}|}
				\hline \textbf{Symbol}&\textbf{Description}&\textbf{Symbol}&\textbf{Description}\\
				\hline
				$\mathcal{M}$, $\mathcal{N}$, $\mathcal{L}$ & Set of D2D transmitters, and D2D requesters and links.
				&$P^{LOS}_{m,n}$ & Probability of LOS link.\\
				$\theta$, $\varphi$ & Antenna angle, antenna beamwidth.
				&$G$, $g$ & Antenna Main-lobe and side-lobe.\\
				$T_{m,n}^A$, $T_P$&Beam alignment time, pilot transmission time.
				&$V_{m,n}$, $\mu_{m,n}$& Relative speed and angle of DT $m$, DR $n$.\\
				$\Delta\mu_{m,n}$, $\alpha$& Misalignment angle and misalignment threshold.
				&$T_{m,n}^S$, $\psi_m$ & Link stability time, sector-level beamwidth.\\
				$T_{PA}$, $T_{BM}$& Peer association time, beamwidth management time.
				&$T_{R}$, $T_{D}$&Time for: PDB-reply, decision making.\\
				$C_m^p$, $R_n^p$& Content $p$: cached in DT $m$, requested by DR $n$.
				&$T_{ACK}$& Time for acknowledgment.\\
				$r_{m,n}$, $\xi_{m,n}$& Data rate, data throughput.
				&$U^{PA}_{n,m}$ & Utility DR $n$ achieve by matching with DT $m$.\\
				$\beta$& Blockage parameter.
				&$U_l$, $u_l$&Aggregate utility of link, individual utility.\\
				$\mathcal{H}_l$ & Set of neighboring D2D links of link $l$.
				&$I_T$, $D_T$&Interference threshold, coverage area.\\				
				$\mathcal{G}_b$& Beamwidth selection game.
				&$P(\mathcal{A})$&Set of probability distributions over $\mathcal{A}_l$.\\					
				$\pi _{\mathbf{a}}$& Steady-state of the game.
				&$\Omega$& set of selected beamwidth.\\
				$T^S_{\max}$, $\delta^p_{\max}$ & Normalization factors.
				&$\mathcal{M}(n)$& Set of feasible DTs for DR $n$.\\
				$f$ & Number of trials in Algorithm \ref{alg:peerAssoc}
				&$T_{max}$&Maximum number of trials in Algorithm \ref{alg:beamwidth}.\\
			\hline
			\end{tabular}
		}
	\end{minipage}\hfill
\end{table*}

\begin{figure}
\centering
\includegraphics[width=1\linewidth, trim={1cm 1cm  .8cm  12cm },clip]{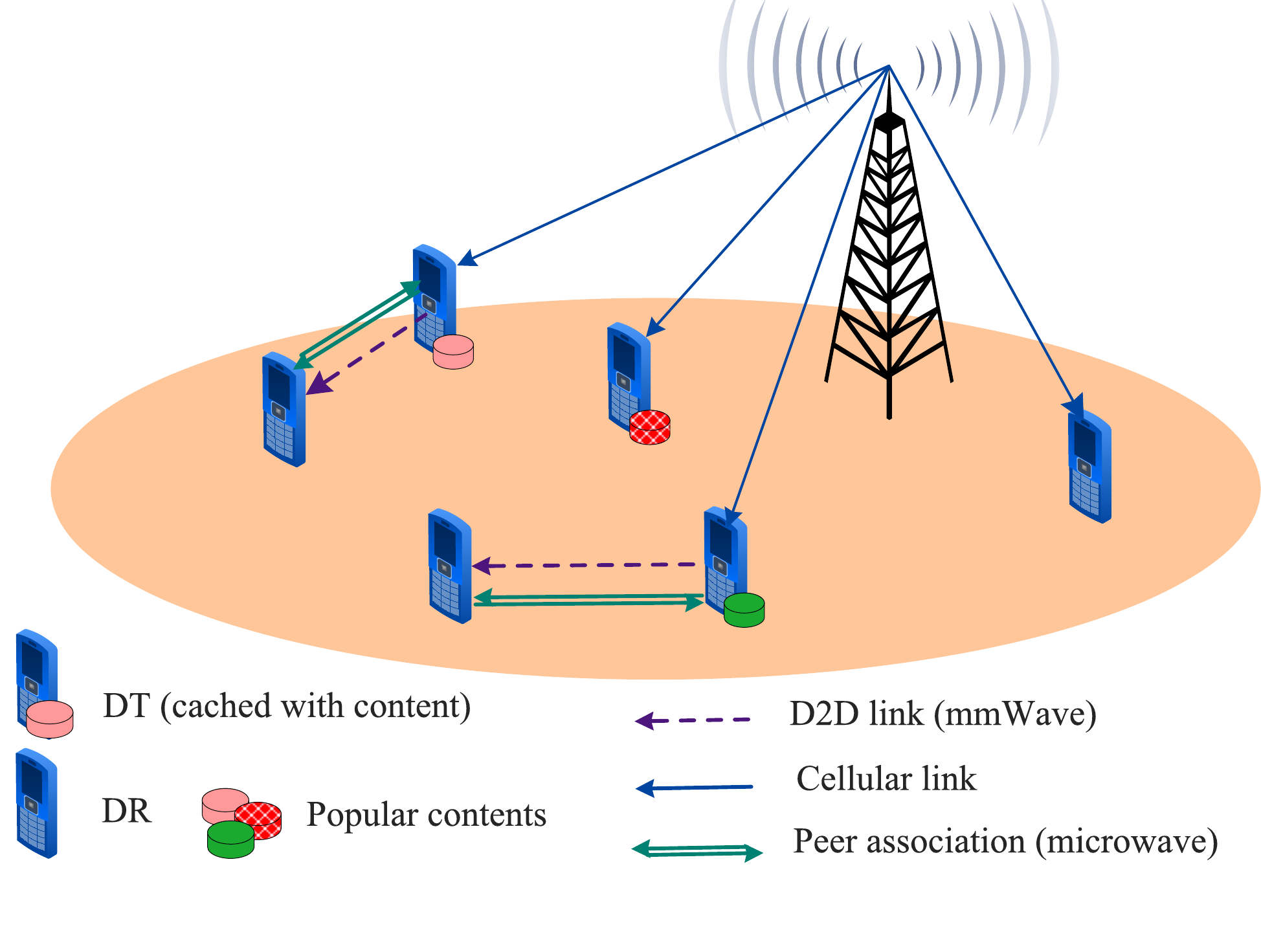}
\caption{Network schematic: DRs retrieve popular contents through D2D links from nearby DTs. The contents are cached in DTs. If the desired content is not found in nearby DTs, DR retrieves the content via cellular links from BS.}
\label{fig:scheme}
\end{figure}

\section{System Model and Assumptions} \label{sec:systemModel}

This section elaborates on system model specifications. In addition, some parameters are defined, such as link stability time, beam alignment overhead, and content segment availability, that will be used later in Section \ref{sec:ProposedScheme}.

\subsection{Network Topology}\label{sec:topology}
\hl{Consider a CCN-based D2D network overlaid the cellular network. The CCN-based D2D network is composed of two entities, namely, D2D transmitters (DTs) and D2D requesters (DRs). Using the CCN protocol DTs are cached with the network's popular contents, while DRs retrieve the popular contents from nearby DTs, as shown in Figure \ref{fig:scheme}.
The data packets are transmitted through establishing D2D links in the mmWave band, operating under time division duplexing (TDD). However, since mmWave band signals are susceptible to blockage, the peer association is performed in dedicated channels in the microwave band.} Let $\mathcal{M}=\{1,..., M\}$ and $\mathcal{N}=\{1,..., N\}$ denote the set of DTs and DRs in the network, respectively. In this system model, D2D communication is used to offload the cellular network, however, if DRs cannot retrieve their desired contents through D2D links, they switch to the cellular network. The details of this process will be discussed in Section \ref{sec:peer association}.
\begin{figure}
\centering
\includegraphics[width=.45\linewidth, trim={11cm 9.7cm  12cm  10cm },clip]{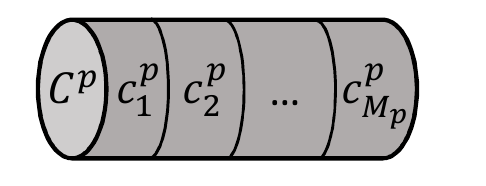}
\caption{Example of a popular content $C^p$ with $M_p$ segments.}
\label{fig:data}
\end{figure}

\subsection{Content and Request Structure}
In the CCN-based network, contents are identified by their name. \hl{Popular contents are cached in DTs' unused storage and are returned in response to the interest message transmitted by DRs.} Each content is fragmented into multiple segments and each segment is addressable. For example, as shown in Figure \ref{fig:data}, we denote a content object $C^p$ cached in DT $m\in \mathcal{M}$ as $C_m^p=\{c^p_{1},..., c^p_{M_p}\}$, where $M_p$ is the number of segments in the content $p$.
$R_n^p=\{r^p_{1},..., r^p_{M_p}\}$
\hl{DRs request the content by transmitting an interest message which includes the name of the desired data packet (instead of the host address).}
Note that the caching policy and \hl{collaboration incentive} are determined by the network hypervisor with the goal to maximize the traffic offloaded from the backhaul network \cite{wang2016hybrid}. In this work, we assume that the popular contents have already been cached in the DTs' storage.
\subsection{Content Segment Availability}
The number of available segments of content $p$ provided by DT $m$ to the DR $n$ can be defined as
\begin{equation}
    \delta_{m,n}^p = |\hspace{.5mm}C^p_m \cap R^p_n\hspace{.5mm}|, \label{eq:conten}
\end{equation}
where $|.|$ denotes the set cardinality. $C_m^p$ and $R_n^p$ represent the set of segments of content $p$, cached in DT $m$ and requested by DR $n$, respectively.
\subsection{Channel Model}
\hl{ To model the mmWave channel, the distance-dependent path-loss model for peer-to-peer communication proposed in \mbox{\cite{rappaport2015wideband}} is adopted.} Under this model the path-loss is defined as $PL(d_{m,n}) = C\hspace{.5 mm}d_{m,n}^{-\alpha}$, where $C$ symbolizes the path-loss intercept, ${\alpha}$ is the path-loss exponent, and $d_{m,n}$ represents the distance between DT $m$ and DR $n$.
Each communication link experiences i.i.d small-scale Nakagami fading with parameter $N_h$. Hence, the received signal power can be modeled as gamma random variable with parameter, $h_{m,n}\sim \Gamma (N_h, 1/N_h)$.
\subsection{Antenna Pattern}
All D2D devices are enabled to perform adaptive directional beamforming in the mmWave band. Beamforming enables users to steer their antenna main-lobe toward the desired direction, as well as adjusting the antenna main-lobe width. Each D2D user can pick a beamwidth from the set of its available beamwidths, $\Phi_i=\{1,...,\varphi_i\}$ for $i \in \mathcal{M}\cup \mathcal{N}$.
\hl{The directional antenna pattern is modeled using the Gaussian model \mbox{\cite{yang2017performance}} as }
\begin{equation}\label{eq:antenna}
G(\theta)\hspace{-1mm}=\hspace{-1mm}\left\lbrace \hspace{-0.5mm}
\begin{array}{ll}
\hspace{-1.2mm}G e^{ -\rho \hspace{1 mm}\theta^2 } \text{,} \hspace{-1mm}&\hspace{-1mm} |\theta|\leq \varphi,\\

\hspace{-1.2mm}g,\hspace{-1mm}&\hspace{-1mm}\text{otherwise,}
\end{array}	\right.
\end{equation}
where $\rho = \frac{2.028 \ln(10)}{\varphi^2}$ and $2\hspace{.5mm}\varphi$ is the antenna half-power beamwidth. $\theta$ denotes the antenna angle
relative to the antenna’s bore-sight direction. \mbox{$G = \frac{\pi 10^{2.028}}{42.64\hspace{.5mm}\varphi\hspace{1mm}+\hspace{1mm}\pi}$} and \mbox{$g = 10^{-2.028}\hspace{.5mm}G$} are the maximum main-lobe gain and the side-lobe gain, respectively \mbox{\cite{yang2017performance}}.
\begin{figure}[t]
    \centering
    \includegraphics[width=.53\linewidth, trim={10.5cm 9.8cm  11.5cm  8.3cm },clip]{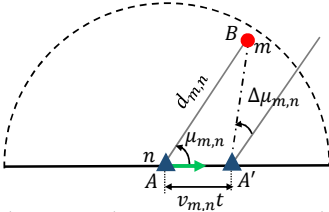}
    \caption{The blue triangle represents DR $n$ located at point $A$ and the red circle represents DT $m$ located at point $B$. The green arrow shows the relative trajectory of DR $n$.}
    \label{fig:linkStability}
\end{figure}

\subsection{Link Stability Time}

A D2D link is stable and proper for data transmission as long as its D2D transmitter and receiver stay aligned. Misalignment in directional communication, due to the users' mobility, occurs when the received power cause drops less than a certain ratio, denoted by $\alpha \in [0,1]$.

Consider a given D2D link whose receiver and transmitter are located at point $A$ and $B$, respectively, as shown in Figure \ref{fig:linkStability}. Assume that the transceivers' antenna beams are aligned.
Also, the receiver is moving with relative velocity $V_{m,n}$ in the direction of the relative angle of $\mu_{m,n}$ (with respect to its antenna bore-sight direction). Since the bore-sight angle of D2D transceivers is fixed, the movement will cause beam misalignment. The pointing error of the D2D receiver toward its transmitter $\Delta t$ seconds later, denoted by $\Delta \mu_{m,n}$, can be obtained using the law of sines in triangle $ABA'$ as
\begin{equation}
   \frac{\sin (\Delta \mu_{m,n})}{V_{m,n}\Delta t} =\frac{\sin (\mu_{m,n})}{d_{m,n}},\nonumber
\end{equation}
where $d_{m,n}$ denotes the D2D links distance. Note that although mobility changes $d_{m,n}$, the impact of distance difference is neglected and only the impact of movement on the angular difference is considered. Also, we assume that $V_{m,n}\hspace{.5mm} \Delta t \hspace{.5mm} \ll\hspace{.5mm}  d_{m,n}$. For small $\Delta \mu_{m,n}$, we estimate $\sin(\Delta \mu_{m,n})$ as $\sin(\Delta \mu_{m,n}) \simeq \Delta \mu_{m,n}$, therefore,
\begin{equation}
    \Delta \mu_{m,n} \simeq \frac{V_{m,n}\hspace{.75mm}\Delta t\hspace{.75mm}\sin (\mu_{m,n})}{d_{m,n}}.\label{eq:delta_mu}
\end{equation}

Based on the definition, the link is stable if the relative antenna gain is above a certain ratio, $\alpha \in [0,1]$.
\begin{equation}
     \alpha =\frac{G(\theta= \Delta \mu_{m,n})}{G(\theta=0)}= e^{-\rho \hspace{1 mm}\Delta  \mu_{m,n}^2},\label{eq:alpha}
\end{equation}
Using \eqref{eq:delta_mu}  and \eqref{eq:alpha} the link stability time, denoted by $T_{m,n}^S$, can be written as
\begin{equation}
    T^S_{m,n} = \frac{d_{m,n}\hspace{.5mm}\varphi_n}{V_{m,n} \hspace{.5mm} \sin (\mu_{m,n})}\sqrt{\frac{\ln (\frac{1}{\alpha})}{2.028\hspace{1mm}\ln  (10)}}.\label{eq:linkS}
\end{equation}
It can be seen that higher antenna beamwidth and lower gain threshold increase the link stability. Moreover, lower relative speed guarantees D2D links to be stable for longer.

\subsection{Beam Alignment Overhead}
Beam alignment between DT-DR requires sending and receiving multiple pilot signals. In this work, the hierarchical beam alignment method is used, where first the best wide-beam pair is found through an exhaustive search, and then the search is refined using a narrower beam level within the subspace of the wide-beam pair \cite{7914759}.
Assuming the antenna wide-beam pairs are already aligned, the narrow-beam alignment time \cite{shokri2015beam} can be written as
\begin{equation}
T_{m,n}^A =\left \lceil{\frac{\psi_m}{\varphi_m}}\right\rceil \left\lceil{\frac{\psi_n}{\varphi_n}}\right\rceil T_P,\label{eq:alignD}
\end{equation}
in which $\psi_i$ and $\varphi_i$ denote the wide- and narrow- level beamwidth of D2D user $i$, respectively.
$T_P$ represents the pilot signal transmission time. Note that although narrower antenna beamwidth provides higher antenna gain based on \eqref{eq:antenna}, it increases beam alignment overhead based on \eqref{eq:alignD}.
Since beam alignment time must be less than the link stability time, $T^S_{m,n}$, the lower bound on feasible beamwidths can be derived as
\begin{equation}
\varphi_m \varphi_n\geq \psi_m\psi_n \hspace{3pt}\frac{ T_P}{T^S_{m,n}}\label{eq:lowerBoundBeam}.
\end{equation}

On the other hand, the antenna beamwidth cannot be higher than wide-level beamwidth. Therefore, the upper bound of antenna beamwidth can be written as $\varphi_m \leq \psi_m$ and $\varphi_n \leq \psi_n$.

\subsection{Timing and Signaling Structure}
Due to the lack of central controller D2D  users are required to initialize the communication through signaling. It is worth noting that initialization is performed through the common control channel (CCC) in microwave bands.
As shown in Figure \ref{fig:frame}, initialization includes peer association and beam management with a duration of $T_{PA}$ and $T_{BM}$, and must be performed prior to data transmission.

Peer association with duration of $T_{PA}$ starts with a DR broadcasting the peer discovery beacon (PDB). It consists of three phases, a) PDB-reply with the duration of $T_R$, where DTs reply back to the PDB signal with their context information, b) decision making with duration $T_D$, when DRs decide on their fittest DT, and c) acknowledgment with duration $T_{\text{ACK}}$. The details of the peer association algorithm will be discussed further in Section \ref{sec:peer association}.

Following the peer association, a given matched DT-DR $(m,n)$ has a time duration of $T_{m,n}^S$ to perform the beam management and transmit data. Beam management is implemented for aligning antenna beam and selecting the antenna beamwidth for data transmission. The beam management time $T_{BM}$ which consists of $T_{m,n}^A$ and $T_{BWS}$ depends mainly on the selected antenna beamwidth as per \eqref{eq:alignD}. The beamwidth selection algorithm and the trade-off between antenna beamwidth and data throughput will be discussed in-depth in Section \ref{sec:game}.

\begin{figure}[t]
	\centering
	\includegraphics[width=1\linewidth, trim={7.5cm 7cm  8.2cm  9.6cm },clip]{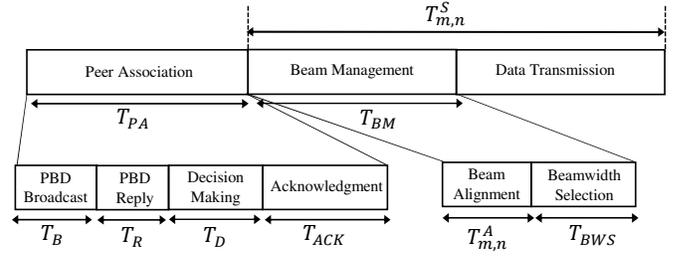}
	\caption{CCN-based D2D communication frame structure consists of peer association, beam management and data transmission slots. Peer association duration is fixed, while beam management duration and data transmission time  depend on the antenna beamwidth.  }
	\label{fig:frame}
\end{figure}

\section{Problem Formulation}\label{sec:problemFormulation}
Establishing directional D2D links in the mmWave band requires two crucial steps, i.e., peer associations and beamwidth selection. In this section, we elaborate on formulating the sum data throughput maximization problem in the CCN-based D2D network with respect to peer association and antenna beamwidth selection.

Let $\mathcal{A}=\{a_{m,n}: m\in\mathcal{M},\:n\in\mathcal{N}\}$ denotes the set of all possible DT-DR associations in the network, where $a_{m,n}$ as the association parameter can be defined as
\begin{equation}
a_{m,n}=\left\lbrace
\begin{array}{ll}
1, &\text{if link between $m$ and $n$ exists,}\\
0, &\text{otherwise.}\nonumber\\
\end{array}
\right.
\end{equation}

We future define $\mathcal{M}(n)\subseteq \mathcal{M}$ as the
subsets of feasible DTs for DR $n$, and the feasibility conditions are defined as
 \begin{equation}
     \mathcal{M}(n) = \big \{m \in \mathcal{M} \hspace{4 pt} \big| \hspace{4 pt} d_{m,n} \leq d_{T}, \hspace{4 pt} \sigma_{m,n}=1  \big \},\nonumber
 \end{equation}
where $d_{m,n} $ and $d_{T}$ represent the Euclidean distance between DT $m$ and DR $n$, and the coverage distance of DR $n$, respectively.
Since mmWave band communication undergoes severe attenuation in non-line-of-sight (NLOS) links \cite{rappaport2013millimeter}, we assume that a D2D link can be established only when the LOS link exists. Parameter $\sigma_{m,n}$ captures the impact of blockage and is modeled as a Bernoulli random variable
\begin{equation}
\sigma_{m,n}=\left\lbrace
\begin{array}{ll}
1, &\text{if LOS exists},\\
0, &\text{otherwise.}\nonumber\\
\end{array}
\right.
\end{equation}
\hl{It is worth noting that using a random Boolean scheme of rectangles to model the blockages, a link of length $d_{m,n}$ is LOS with probability $P^\text{\tiny LoS}_{m,n}=\exp(-\beta \hspace{.5mm} d_{m,n})$, where parameter $\beta$ depends on the average size and density of blockages \mbox{\cite{bai2014analysis}}}.
\hl{Users can estimate the location of neighboring users \mbox{\cite{nitsche2015steering}}, and determine whether the link is LOS \mbox{\cite{bahadori2018device}} using angle of arrival (AoA) spectrum method.}

Let the set of selected beamwidths of associated DT-DR pairs for data transmission denotes as
\begin{equation}
    \Omega = \left\{ (\varphi_{m}, \varphi_{n})\hspace{4 pt}\big|\hspace{4 pt} n\in\mathcal{N},\:m\in\mathcal{M}(n), \hspace{.5mm} a_{m,n}=1\right\}.\nonumber
\end{equation}

The achievable data-rate on a given D2D link between DT $m$ and DR $n$, which depends on the set of paired D2D users $\mathcal{A}$, as well as their antenna beamwidth $\Omega$ can be defined as
\begin{equation}
r_{m,n}\left(\mathcal{A},\Omega\right)=\gamma \hspace{.5mm} B\log_2\left(1+\text{SINR}_{n}(\mathcal{A},\Omega)\right), \mbox{}\nonumber
\end{equation}
where $\gamma = \left(1-\frac{T^A_{m,n}}{T^S_{m,n}}\right)$ captures the impact of beam alignment overhead, and $B$ denotes the available bandwidth. The achieved signal-to-noise-plus-interference-ratio (SINR) on DR $n$ is denoted by $\text{SINR}_{n}$.
\hl{Moreover, The data throughput of a given D2D link which is defined as the amount of the data that is transmitted on the link during the link stability time \mbox{$T^S_{m,n}$}, can be written as}
\begin{equation}
    \xi_{m,n} \left(\mathcal{A},\Omega\right) = \frac{r_{m,n} \left(\mathcal{A},\Omega\right) \times T^S_{m,n}}{\delta^p_{m,n}}.\label{eq:dataThrou}
\end{equation}

\hl{The problem addressed in this work can be formulated
as designing a peer association and antenna beamwidth selection algorithms such that they maximize the network sum-throughput as}
\begin{subequations} \label{eq:mainOptprb}
	\begin{align}
		&\underset{\mathcal{A},\Omega}{\text{Maximize}}&&\hspace{-2.2mm}\sum\limits_{m\in\mathcal{M}}\sum\limits_{n\in\mathcal{N}} a_{m,n} \hspace{1mm}  \xi_{m,n} \label{eq:fitness}\\
		& \quad\;\text{subject to:} &&
		\hspace{-2.2mm}\sum\limits_{n\in \mathcal{N}}a_{m,n}=1, \hspace{4pt}\forall m \in \mathcal{M}, \label{eq:matchres1}\\
		&&&\hspace{-2.2mm}\sum\limits_{m\in \mathcal{M}}a_{m,n}=1,  \hspace{4pt}\forall n \in \mathcal{N}, \label{eq:matchres2}\\
		&&&\hspace{-2.2mm}a_{m,n}\hspace{-0.8mm} \in\hspace{-0.8mm} \{0,1\},\hspace{4pt} \forall m,n\in \mathcal{M}\hspace{-0.8mm}\times\hspace{-0.8mm}\mathcal{N}, \label{eq:matchres3}\\
		&&&\hspace{-2.2mm}\varphi_m\:\varphi_n\hspace{-1mm}\geq\hspace{-1mm} \psi_m \psi_n \frac{T_P}{T^S_{m,n}}, \hspace{4pt} \forall m,n\in \mathcal{M}\hspace{-0.8mm}\times\hspace{-0.8mm}\mathcal{N},\label{eq:antennaCon}\\
		&&&\hspace{-2.2mm}\varphi_m \leq \psi_m \label{eq:mainOptprb_g},\hspace{4pt} \forall m \in \mathcal{M},\\
		&&&\hspace{-2.2mm}\varphi_n\leq \psi_n \label{eq:mainOptprb_h}, \hspace{4pt}\forall n \in \mathcal{N}
	\end{align}
\end{subequations}
where constraints \eqref{eq:matchres1}-\eqref{eq:matchres3} show that the matching among DTs and DRs must be one to one. Constraints \eqref{eq:antennaCon}-\eqref{eq:mainOptprb_h} represent the antenna beamwidth upper and lower bound according to \eqref{eq:lowerBoundBeam}.
The optimization problem \eqref{eq:mainOptprb} is an NP-hard combinatorial optimization problem, which can be solved using centralized exhaustive search algorithms. However, utilizing the assistance of the central controller (base station) is against the main purpose of implementing D2D communication for offloading the network. Restricted access to global information in our application motivates us to seek for low-complexity and low-overhead algorithms that enable D2D users to pick strategies that maximize the network data throughput.

\vspace{-5pt}
\section{Proposed Scheme}\label{sec:ProposedScheme}

For analytical tractability, the optimization problem in \eqref{eq:mainOptprb} is decomposed into two separate problems, namely, peer association and beamwidth selection.
The optimization problem \eqref{eq:mainOptprb} with constraints \eqref{eq:matchres1}-\eqref{eq:matchres3} denotes as a peer association problem. A classic approach to solve this problem is modeling it as a matching game with a well-defined utility function \cite{  perfecto2017millimeter, namvar2015context}, to which deferred acceptance algorithm (DAA) provides a polynomial-time converging solution \cite{gale1962college}. DAA obtains the stable mapping among two sets of D2D links with size $n$ (DTs and DRs),  given an ordering of preferences for each D2D link.
However, the DAA's convergence is time-consuming.
To reduce the signaling overhead and initialization delay, we propose a heuristic algorithm for pairing D2D user, which saves D2D users computational time resources to select the optimal antenna beamwidth.

The optimization problem \eqref{eq:mainOptprb} with constraints \eqref{eq:antennaCon}-\eqref{eq:mainOptprb_h} denotes as the antenna beamwidth selection problem.
The problem of antenna beamwidth selection is modeled as a strategic game, which is proved to be an exact potential game. Further, we show in Section \ref{sec:game} that this problem is guaranteed to have at least one Nash equilibrium. A synchronous Log-linear learning (LLL) based algorithm is proposed to obtain the optimal Nash equilibrium of the game. Thanks to the short-range and directional communication in the mmWave band \cite{rappaport2013millimeter}, only neighboring users are required to exchange information, which accelerates the convergence of the algorithm significantly.
\hl{Note that the dependency between these two problems is relaxed as its exact identification depends highly on network technology. The extensive simulation results in Section \mbox{\ref{sec:NumericalResults}} shows that despite relaxing the interdependency of these problems, implementing the proposed scheme still improves network data throughput significantly compared to existing approaches. Therefore, the interdependency can be relaxed for tractability.}

\subsection{Decentralized Heuristic Algorithm for Peer Association}\label{sec:peer association}

We propose a decentralized, high-speed peer association algorithm with a low computational load that enables DRs to retrieve their desired content from neighboring DTs through stable D2D links. Implementing such an algorithm reduces the initialization overhead considerably. We assume that a) all DTs are willing to share content within their storage unconditionally, and b) the peer association is initiated and decided by DRs.
\hl{In our CCN-based D2D network, peer association starts with a DR broadcasting a peer discovery beacons (PDB) over the common control channel (CCC) that includes $R^p_n$ to indicate its interest in a specific content.}

To select the proper DT, DRs must be able to determine the utility that is gained by matching with the neighboring DTs. The utility function of the DR $n$ through matching with DT $m \in \mathcal{M}(n)$ can be defined as
\begin{equation}
U^{PA}_{n,m}= \frac{T_{m,n}^S}{T^S_{\max}}+\frac{\delta_{m,n}^p}{\delta^p_{\max}},\label{eq:utilityPeer}
\end{equation}
where $T^S_{\max}$ and $\delta^p_{\max}$ are normalization factors, which can be predefined by the users or network. The first term in the utility function captures the stability of the communication link, while the second term indicates the content segment availability. Under the same link stability time, DR $n$ is encouraged to select the DT that provides a higher number of segments of its desires content. It is worth noting that to calculate the $T^S_{m,n}$, all users assume the narrowest antenna beamwidth is implemented. This assumption will be later adjusted in Section \ref{sec:game}.
The details of the D2D peer association in the CCN-based D2D network in mmWave band is given in Algorithm \ref{alg:peerAssoc}.
\begin{itemize}
    \item \hl{ Each DR broadcasts a PDB signal over CCC in the microwave band including the name of the data packet and the desired data segments, $R_n^p$. DR waits for the duration of $T_R$ to receive a response from their neighboring DTs.} Neighboring DTs that received PDB calculate the content segment availability $\delta_{m,n}^p$ as per \eqref{eq:conten}, and packet it along with its context information including geographical position, velocity, and moving direction into PDB-reply and send it back to the DR (lines \ref{algLine:4}-\ref{algLine:5}).
    \item During $T_D$, the DRs that received PDB-reply, calculate the utility that is achieved from matching with each responding DT using \eqref{eq:utilityPeer}, select the fittest DT, and broadcast its decision to its neighbors. DR waits for $T_{\text{ACK}}$ to receive an acknowledgment (ACK) signal from the corresponding DT. In case DR receives ACK signal, the matching is announced by sending back another an ACK signal (lines \ref{algLine:7}-\ref{algLine:11}).
    \item \hl{In case DR has not received any response during $T_R$ or $T_{\text{ACK}}$, failure counter $t_n$ is incremented. If DR fails to retrieve its content from its neighboring DTs for $f$ consecutive times, it will attempt to retrieve its desired content from the cellular network \mbox{(lines \ref{algLine:fail}-\ref{algLine:switch})}.}
    \item At the end of each time frame, $R_n^p$ is updated, and in the case the desired content is not fully received, DR attempts to request for the remaining segments (lines \ref{alg:update}-\ref{alg:content}).
\end{itemize}
\hl{Note that number is trials $f$ can be defined by the D2D users according to the application.}
The proposed algorithm enables DRs to make the association decision with constant computation load. The outcome of the proposed algorithm is the set of established DT-DR links, denoted by $\mathcal{L}$.

\setlength{\textfloatsep}{0.1cm}
\begin{algorithm}[t!]
\small
\SetAlgoLined
\DontPrintSemicolon
\KwResult{$\mathcal{L}$: matching $\mathcal{N}$ and $\mathcal{M}(n)$, $\forall n \in \mathcal{N}$}
 \textbf{Initialize}: ($t_n$, $T_{B}$, $T_{R}$, $T_{D}$, $T_{\text{{ACK}}}$) to zero\; \label{line:init}
 \ForEach{$n \in \mathcal{N}$}{
 \eIf{$t_n\leq f$}{\label{step:check}
      \nonl  \textbf{Phase I: Content discovery}\;

- $n$ broadcasts PDB, count down $T_R$ and wait to receive PDB-reply(s).\;\label{algLine:4}
- $m \in \mathcal{M}(n)$ computes $\delta_{m,n}^p$ as per \eqref{eq:conten} and packet it along with its context information into PDB-reply.\;\label{algLine:5}
- $n$ calculates $U_{n,m}$ as per \eqref{eq:utilityPeer}, selects the fittest DT and announces its decision to its neighbors.\;\label{algLine:7}
\nonl \textbf{Phase II: Link establishment}\;
- $m \in \mathcal{M}(n)$ counts down $T_D$ and wait for $n$'s decision. \; \label{algLine:8}
- Selected DT replies with decision ACK. \; \label{algLine:9}
- $n$ counts down timer $T_{\text{ACK}}$ and wait to receive ACK signal.\; \label{algLine:10}
- $n$ sends back an ACK signal and announces the matching.\; \label{algLine:11}
 \If {$(T_R$ \text{\upshape or} $ T_{\text{ACK}})=0$ $\text{\upshape and no response is received }$}{$t_n\hspace{3pt}=$\hspace{3pt}$t_n+1$, and go to step \ref{step:check}.\label{algLine:fail}}
}{DR $n$ switch to cellular mode and retrieve the content from cellular network.\label{algLine:switch}}
-Update $R_n^p$ and $t_n=0$.\label{alg:update}\;
 \If{ $R_n^p \neq \varnothing$}{go to step \ref{line:init}.}\label{alg:content} }
 \caption{Heuristic D2D peer association}\label{alg:peerAssoc}
\end{algorithm}

\setlength{\floatsep}{0.1cm}
\subsection{Decentralized Algorithm for Optimal Beamwidth Selection}\label{sec:game}
Following the peer association phase, our goal is to enable the associated DR-DT pairs $l \in \mathcal{L}$ to select the proper antenna beamwidth for communication such that the beamwidth strategy profile of all users in the network maximizes the network's sum-data throughput.
As mentioned before, there exists a trade-off between antenna beamwidth and the achievable throughput. Picking a narrower antenna beam, although it leads to higher antenna gain based on \eqref{eq:antenna}, it incurs longer beam alignment overhead based on \eqref{eq:alignD}. Consequently, data transmission time and data throughput is reduced as per \eqref{eq:dataThrou}. In addition, narrower antenna beamwidth leads to lower link stability time, according to \eqref{fig:linkStability}. Therefore, one needs to optimize the antenna beamwidth prior to data transmission according to network conditions and context information.

The beamwidth selection strategy of a D2D link impacts  not only its achievable data throughput but other links' in the network based on \eqref{eq:dataThrou}. Thus, a selfish strategy selection that solely maximizes each link's data throughput cannot guarantee to obtain global optimization. This motivates us to model the beamwidth selection problem as a strategic game to consider the interaction among users and to obtain the optimal beamwidth strategy for all D2D links.
Recently, game theory has been used to model the interaction among users in wireless networks, for example, for beam-pair selection \cite{liu2018decentralized}, and channel selection \cite{xu2011opportunistic}.
Information exchange among users is crucial in such games to reach global optimization. In our decentralized system model, this may lead to significant signaling overhead and delay in communication.
However, directional transmission and the short range of communication in the mmWave band reduces the information exchange only to neighboring (i.e., interacting) links instead of the whole network, which in turn reduces the overhead significantly.
For a given established D2D link $l\in \mathcal{L}$, between DR $n$ and DT $m$, the set of interacting D2D links are those whose DTs are a) in the coverage area of the DR $n$, and b) causes interference on DR $n$ that is higher than a predefined threshold $I_T$.
\begin{figure}[t!]
    \centering
    \vspace{-.5cm}
    \includegraphics[width=.6\linewidth, trim={10cm 8.5cm  10cm  6cm },clip]{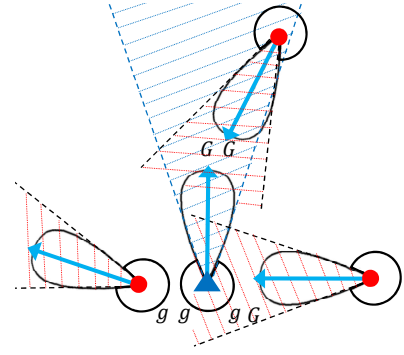}
    \caption{Directional interference: red circles and the blue triangle show the DTs and DR, respectively. $G$ and $g$ denote the main-lobe gain and side-lobe gain of the antenna.}
    \label{fig:directionalInter}
\end{figure}
The set of interacting D2D links of a given D2D link $l$, denoted by $\mathcal{H}_l$ can be defined as
\begin{equation}
     \mathcal{H}_l=\{i\in \mathcal{M} \big | \hspace{1mm}d_{n,i}\leq d_{T}, \hspace{1mm} I_{i} \geq I_{T}\}\label{eq:interaction}
\end{equation}
where $d_{n,i}$ and $d_{T}$  represent the Euclidean distance between DR $n$ and DT $i$ and the coverage distance of DR $n$, respectively.
$I_{i}$ is the directional \hl{LOS} interference caused by DT $i$ on DR $n$. The received interference power is related to the transmission power, the channel gain, and the transmission/reception antenna gain. We assume the same transmission power for all DTs, and also equal channel gains for all D2D links within the same transmission range. Therefore, the received interference is mainly determined by transmission/reception antenna gains which depend on the antenna beamwidth and bore-sight direction, based on \eqref{eq:antenna}. \hl{Note that the impact of NLOS interference is negligible thus only LOS interference is considered \mbox{\cite{thornburg2016performance}}.}

As shown in Figure \ref{fig:directionalInter}, the effective antenna gain of interference can be determined using the antenna boresight direction of the DR and its neighboring DTs. The interference gain, denoted by $g_I$, can be a) $GG$ when both DR and DT have their main-lobes directed towards each other, b) $Gg$ when DR has its main-lobe directed to side-lobe of DT or vice versa, and c) $gg$ when both DT and DR have their side-lobes aligned. Therefore, if DR $n$ knows the location, bore-sight direction and beamwidth of a DT $i$, it can simply calculate the amount of the received directional interference using $I_i = P  h_{i,n}g_I\text{PL}(d_{i,n})$.

\vspace{2mm}
\subsubsection{Game Model}
The beamwidth selection game is denoted as $\mathcal{G}_b = \big[\mathcal{L},\{\mathcal{A}_l\}_{l \in \mathcal{L}}, \{U_l\}_{l \in \mathcal{L}}\big ]$, where $\mathcal{L}$ is the set of players (DT-DR links that are established using Algorithm \ref{alg:peerAssoc}), $a_l \in \mathcal{A}_l$ denotes the strategy of player $l$ (beamwidth), where $\mathcal{A}_l$ is the set of available actions for player $l$ ($\Phi_l$), and $U_l$ represents the utility of player $l$. For simplicity and without loss of generality, we assume that DTs and their corresponding DRs adopt the same beamwidth strategy. This can be extended to the case that users implement different strategies.
$\mathbf{a}_{-l} \in \mathcal{A}_{-l}$ denotes the beamwidth selection profile of all the players excluding player $l$, and $\mathcal{A}_{-l} = \times_{i \in \mathcal{L}\backslash l } \mathcal{A}_{i}$ is the joint strategy space of all the players excluding player $l$, and $\times$ denotes the Cartesian product.

The utility function of D2D link\footnote{Hereafter, we refer to the matched D2D pairs through Algorithm \ref{alg:peerAssoc} as "D2D links" and use subscript $l$ instead of $m,n$.} $l$ which depends on the utility of itself and the utility of its neighboring links can be defined as
 \begin{equation}
	U_l(a_l,\mathbf{a}_{-l})=	u_l(a_l, \mathbf{a}_{\mathcal{H}_l})+\sum_{i \in \mathcal{H}_l}	u_i(a_i,\mathbf{a}_{\mathcal{H}_i})\label{eq:totalUt}
\end{equation}
with
\begin{equation}
u_l(a_l,\mathbf{a}_{\mathcal{H}_l}) = r_l (a_l,\mathbf{a}_{\mathcal{H}_l}) + C \hspace{3pt}\varepsilon \hspace{3pt} \Lambda\left(T_l^\delta,T'\right) \label{eq:indUt}
\end{equation}
where $\mathbf{a}_{\mathcal{H}_l} \in \mathcal{A}_{\mathcal{H}_l} = \times_{i \in \mathcal{H}_l}\mathcal{A}_{i}$ is the set of joint beamwidth selection strategy of neighboring D2D links of D2D link $l$, $\Lambda(x,y)$ is a binary function that is $-1$ if $x > y$ and $0$ otherwise. $T_l^\delta= \frac{\delta^p_{l}}{r_{l}}$ represents the time required to transmit data segments $\delta^p_l$, and $T' = T^S_l-T^A_l$ denotes the time remained to transmit the data after beam alignment. Finally, $C$ and $\varepsilon =  |1-\frac{T'}{T_l^\delta}|$ represents the penalty scalar and penalty coefficient, respectively.
The first term in (\ref{eq:totalUt}) captures the individual utility of D2D link $l$, while the second term is the utility of all $l$'s interacting D2D links. Further, the individual utility of each D2D link depends on its data throughput. To capture this, the leftmost term in (\ref{eq:indUt}) denotes the link's data rate, and the rightmost term captures the trade-off between data throughput and antenna beamwidth in \eqref{eq:alignD}.

It can be seen that the defined utility function in \eqref{eq:totalUt} is in line with the optimization problem in \eqref{eq:mainOptprb} with constraints \eqref{eq:antennaCon}-\eqref{eq:mainOptprb_h}. Therefore,
our objective is to find a joint beamwidth selection profile that maximizes the utility of all active D2D links $l \in \mathcal{L}$. The beamwidth selection game can be defined as
\begin{equation}
    \mathcal{G}_b: \max_{a_l \in \mathcal{A}_l} U_l(a_l,\mathbf{a}_{-l}), \hspace{4 mm} \forall l \in \mathcal{L}. \label{eq:game}
\end{equation}

The Nash equilibrium is the stable solution of the strategic games such as $\mathcal{G}_b$ and can be defined as follows.

\begin{definition} Nash equilibrium (NE):
    A beamwidth selection profile $\mathbf{a}^*=(a^*_1,a^*_2,...,a^*_L)$ is a pure strategy NE point if and only if no D2D link can improve its utility by deviating unilaterally, i.e.,
    \begin{equation}
        U_l(a_l^*,\mathbf{a}_{-l})\geq U_l(a_l,\mathbf{a}_{-l}) \hspace{3 mm} \forall l \in \mathcal{L}, \hspace{3 mm} a_l^* \neq a_l
    \end{equation}
\end{definition}

Although NE is the steady-state of strategic games, an important question is whether the beamwidth selection game will reach a steady-state (NE) eventually.
A given utility function may have multiple Nash equilibria or may not have any. \hl{Therefore, it is crucial to verify that at least one NE exits for $\mathcal{G}_b$.}
The properties of NE of beamwidth selection game $\mathcal{G}_b$ are characterized by the following theorems.

\begin{theorem}\label{th:potential}
    Beamwidth selection game $\mathcal{G}_b$ is an exact potential game with potential function $\Theta(a_l,\mathbf{a}_{-l})=\sum_{l \in \mathcal{L}}u_l(a_l,\mathbf{a}_{-l})$, which has at least one pure NE, and optimal solution of the network data throughput maximization problem constitutes a pure strategy NE of $\mathcal{G}_b$.
\end{theorem}
\begin{proof}
    See Appendix \ref{ap:potential}.
\end{proof}

All potential games share the finite improvement property (FIP). According to FIP, letting a player deviate to a better strategy using the best response dynamics, terminates to a NE in a finite number of iterations \cite{monderer1996potential}. As Nash equilibria are the maximizers of the potential function, and the potential function $\Theta$, represents the sum of the network's data throughput, therefore, the optimal beamwidth selection profile that maximizes the network data throughput can be achieved by finding the optimal NE points of the game $\mathcal{G}_b$. However, in case that multiple Nash equilibria exist, an efficient learning algorithm is required to achieve optimal NE.

\vspace{2mm}
\subsubsection{Decentralized Algorithm to obtain the optimal NE}
As it is already mentioned, applying the best response dynamics eventually leads to a pure NE. However, the best response algorithm does not guarantee to converge to the optimal pure NE in case potential function $\Theta$ has multiple optimums. Log-linear learning (LLL) is a classical algorithm that guarantees the convergence to a set of optimal pure NEs of an exact potential game \cite{marden2012revisiting}. The algorithm follows the same procedure as the best response algorithm \cite{monderer1996potential}, however, it allows the possibility of exploration by deviations from the best response with a small probability. In LLL, at each time step ($k$) one randomly chosen player $l$, is allowed to alter its strategy according to a mixed strategy $p_l^{a_l} \in P (\mathcal{A}_l)$, where $P (\mathcal{A}_l)$ is the the set of probability distributions over $\mathcal{A}_l$.
Meanwhile, the actions of all other players remain fixed, i.e., $\mathbf{a}_{-l}(k+1)=\mathbf{a}_{-l}(k)$. D2D link $l$’s mixed strategy is updated based on Boltzmann rule and is given by
\begin{equation}
    p_l^{a_l}(k+1)=\frac{\exp\left(\frac{1}{\tau}U_l(a_l,\mathbf{a}_{-l}(k))\right)}{\sum_{\bar{a}_l\in \mathcal{A}_l}\exp\left(\frac{1}{\tau}U_l(\bar{a}_l,\mathbf{a}_{-l}(k))\right)}\label{eq:Boltzmann}
\end{equation}
where $\tau$ is the learning parameter. For very large $\tau$, the strategy of each D2D link is chosen approximately based on a uniform distribution over its set of actions. While, for small $\tau$, the selected strategy is a uniform distribution over best responses against $\mathbf{a}_{-l}$. It can be seen in \eqref{eq:Boltzmann} that if taking an action leads to gaining higher utility compare to other actions, this action has higher chance of being selected in future. Thus, the algorithm will converge to a network optimum eventually.
\setlength{\textfloatsep}{0.1cm}
\begin{algorithm}[t]
\small
\DontPrintSemicolon
    \textbf{Initialization}: $k=0$, $T_{\max}$, set the strategy profile of all users randomly, $\mathbf{a}(0)=\{a_1(0),a_2(0),...,a_L(0)\}$.\;\label{algLine:2-1}
    \SetKwRepeat{Repeat}{repeat}{until}
    \Repeat{ \;
    $\big |U_l(k+T_{\max})-U_l(k)\big| \simeq 0 $ \upshape or,\; \label{algLine:2-7}
    \upshape There exist a component of $P (\mathcal{A}_l)$  which is sufficiently close to 1.\label{algLine:2-8}}
    {\ForEach{D2D links $l \in \mathcal{L}$}{
    Find the set of neighboring users, $\mathcal{H}_l$, using \eqref{eq:interaction}.\;
    Exchange information with neighboring links.}\label{algLine:2-3}
    A set of non-interacting D2D links are selected randomly denoted by $\tilde{\mathcal{L}} \subset \mathcal{L}$.\; \label{algLine:2-4}
    \ForEach{D2D link $l \in \tilde{\mathcal{L}}$}{
    Calculate the utility over its all available actions, $U_l(k)$, $\forall a_l \in \mathcal{A}_l$ using the information received from its neighbors while $\mathbf{a}_{-l}(k+1)=\mathbf{a}_{-l}(k)$.\; \label{algLine:2-5}
    Select an action according to the mixed strategy vector $P (\mathcal{A}_l)$.\;
    Update the mixed strategy vector $P (\mathcal{A}_l)$ using \eqref{eq:Boltzmann}.} \label{algLine:2-6}
    $k=k+1$\;} \label{algLine:2-end}
    \caption{Synchronous LLL-based beamwidth selection}\label{alg:beamwidth}
\end{algorithm}

LLL requires asynchrony, which refers to the fact that D2D links are allowed to update their strategies one at a time. However, it has been shown that this assumption can be relaxed by allowing a group of non-interacting players to update their strategies simultaneously \cite{marden2012revisiting}.
The short-range and directional communication in the mmWave band allows a group of non-interaction D2D links to update their strategies simultaneously, which in turn expedites the convergence of the algorithm. Further, each D2D link only needs to exchange information with its neighboring D2D links $\mathcal{H}_l$, which reduces the signaling overhead significantly. The details of the synchronous LLL-based beamwidth selection algorithm is given in Algorithm \ref{alg:beamwidth}.
\begin{itemize}
    \item Each D2D link selects an initial beamwidth based on a uniform distribution over its available beamwidth set $\Phi_l$. Then determines its neighboring links $\mathcal{H}_l$ using \eqref{eq:interaction} and exchange information including position, bore-sight direction, selected beamwidth and gained utility with neighboring links (lines \ref{algLine:2-1}-\ref{algLine:2-3}).
    \item A set of non-interacting D2D links are selected randomly denoted by $\tilde{\mathcal{L}} \subset \mathcal{L}$. Since  there  is  no  central controller in  the  network,  the users are chosen using  contention mechanisms over CCC \cite{xu2011opportunistic} (line \ref{algLine:2-4}).
    \item Each of selected D2D links calculates its gained utility $U_l(k)$ using \eqref{eq:indUt} over all its possible actions. Then selects a strategy randomly based its mixed strategy distribution $P (\mathcal{A}_l)$ and updates its mixed strategy using \eqref{eq:Boltzmann} (lines \ref{algLine:2-5}-\ref{algLine:2-6}).
    \item This will continue until the utility of a D2D link has not been changed through $T_{\max}$ iterations or one of its mixed strategy elements $p_l^{a_l}$ is sufficiently close to 1 (lines \ref{algLine:2-7}-\ref{algLine:2-8}).
\end{itemize}
The asymptotic behavior of the synchronous LLL-based algorithm, as the iteration
number goes sufficiently large, can be defined using the following theorem.

\begin{theorem} \label{th:station}
    For the beamwidth selection game $\mathcal{G}_b$, if all D2D links $l \in \mathcal{L}$ adhere to the synchronous LLL-based beamwidth selection algorithm, the stationary distribution $\pi _{\mathbf{a}} \in P (\mathcal{A})$ of the joint action profile for $\forall \tau > 0$ converges to
    \begin{equation}
        \pi _{\mathbf{a}} = \frac{\exp\left(\frac{1}{\tau}\Theta (\mathbf{a})\right)}{\sum _{\bar{\mathbf{a}}\in \mathcal{A}}\exp\left(\frac{1}{\tau}\Theta(\bar{\mathbf{a}})\right)}\label{eq:stationary}
    \end{equation}
    where $\mathcal{A}=\times _{i\in \mathcal{L}} \mathcal{A}_i$ and $\Theta (.)$ denotes the potential function.
\end{theorem}
\begin{proof}
    See Appendix \ref{ap:station}
\end{proof}

It can be seen in \eqref{eq:stationary} that as $\tau \rightarrow 0$, $\pi_a \rightarrow 1$, LLL converges to the potential function maximizer with high probability. In case multiple maximizers exist, LLL converges to one of the maximizers with a uniform distribution. Therefore, the LLL-based beamwidth selection algorithm converges to the optimal solution of the $\mathcal{G}_b$ with high probability. It is worth noting that parameter $\tau$ captures the trade-off between the exploration of the beamwidth strategies and the speed of convergence. In practice, it is advised to start with a large $\tau$ and keep decreasing as the process iterates. In Section \ref{sec:NumericalResults}, we choose $\tau=\frac{1}{k}$, where $k$ is the algorithm iteration number.

The decentralized LLL-based beamwidth selection algorithm enables D2D users to select the antenna beamwidth considering context information through limited information exchange with neighboring users, which makes it suitable for large-scale D2D networks.

\subsection{Content Sharing Scenario}
One of the main applications of the proposed initialization scheme is content sharing in the cellular network via D2D links in the mmWave band. Our proposed algorithm implements the CCN protocol for content sharing. In this framework, the popular contents are cached in DTs’ unused memory and will be shared/transmitted to nearby DRs on demand.
Equation \eqref{eq:utilityPeer} guarantees that a DR is associated with a DT that is cached with its desired content segments (first term), and also DT can provide the content via a stable D2D link (second term).
Moreover, Algorithm \ref{alg:beamwidth}) optimizes the users’ antenna beamwidth to maximize the D2D network data throughput and assure successful content delivery/sharing.

However, it is worth noting that content sharing has other important virtues such as optimizing content caching \cite{song2017learning}, incentivizing transmitters \cite{wu2017social} and routing \cite{lee2011proxy,6193514}. We assumed that the popular contents are already identified and cached in the DTs by the network and DTs share the content within their storage unconditionally. To address the content popularity evaluation, optimizing content caching, and motivating DTs, the following scenarios can be considered.

In the first scenario, to consider DTs incentive to share the content, the optimization problem \eqref{eq:mainOptprb} must be modified to include the costs associated with involving DTs in content sharing. To motivate DTs, DRs can pay a small monetary fee that covers the fee charged by the DTs’ network services provider and energy consumption. Also, the network can motivate the DT to share its cached content by getting some monthly service fee discount in return for assisting the network in disseminating data.

In the second scenario, to consider content popularity distribution, network can use data from DTs’ social media to identify the popular contents and evaluate the desire of DTs in sharing contents \cite{song2017learning, wu2017social}. In addition, the network can use this information to locate/identify DTs at the hot spot area to cache the popular content in those users. Also, the network can cache the contents in DTs according to their location \cite{kwak2018hybrid}.  For example, network prioritizes edge DTs as candidates for content caching since DRs closer to BS has better channel condition. However, it would be more efficient for the DRs on the cell edge to retrieve their desired content from nearby DTs rather than the BS. In these cases, Algorithm \ref{alg:peerAssoc} should include \textbf{popular content discovery} and \textbf{content caching} prior to content discovery and link establishment.

\begin{table}[t]
\begin{center}
\caption{Simulation parameters}\label{tab:parameters}
\begin{tabular}{ |c|c|  }
\hline
Parameter & Value \\
\hline
Carrier frequency & 28 GHz\\
Communication bandwidth($B$) & 100 MHz \\
Thermal noise density  & -174 dBm/Hz \\
Transmission power & 15 $dBm$\\
Free space path-loss & -61.7 dB\\
LOS path-loss exponent & 2 \\
Antenna beamwidth($\Phi_l$) & [15:10:45]\\
Distance threshold ($D_{T}$)& 50 m  \\
D2D pairs distance($d_l$) & $\sim \mathcal{U}(30,80)$ \\
penalty parameter ($\eta$)  & $\max(r_l)$  \\
Pilot transmission time ($T_p$)& 10 $\mu$s \\
LoS link prob. ($\beta$) &  0.0027 \\
Number of contents ($N_c$) & 5 \\
Velocity of users ($V$) & $\sim \mathcal{U}(1,3)$ $mph$\\
Moving direction ($\theta$) & $\sim \mathcal{U}(-\pi,\pi)$\\
\hline
\end{tabular}
\end{center}
\end{table}

\begin{figure*}[ht]
\centering
\vspace{-.6cm}
	\captionsetup[subfigure]{aboveskip=1pt}
	\begin{tabular}{@{}c@{}}
		\subfloat[]{
			\includegraphics[scale=.55,trim={4.cm 8.5cm  4.5cm  9cm },clip]{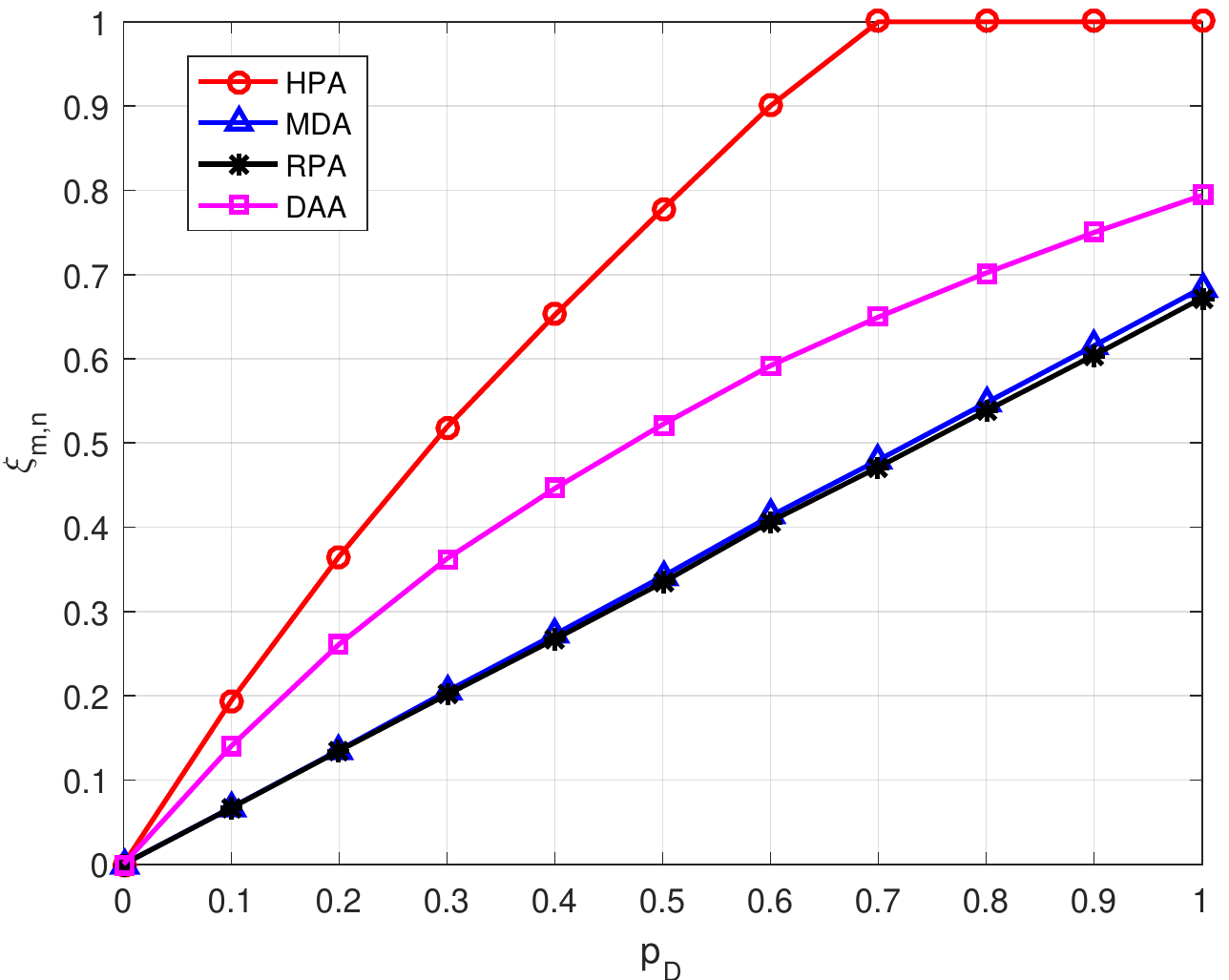}\label{fig:PAprob}}
		\subfloat[]{
			\includegraphics[scale=.55,trim={4cm 8.5cm  4.5cm  9cm },clip]{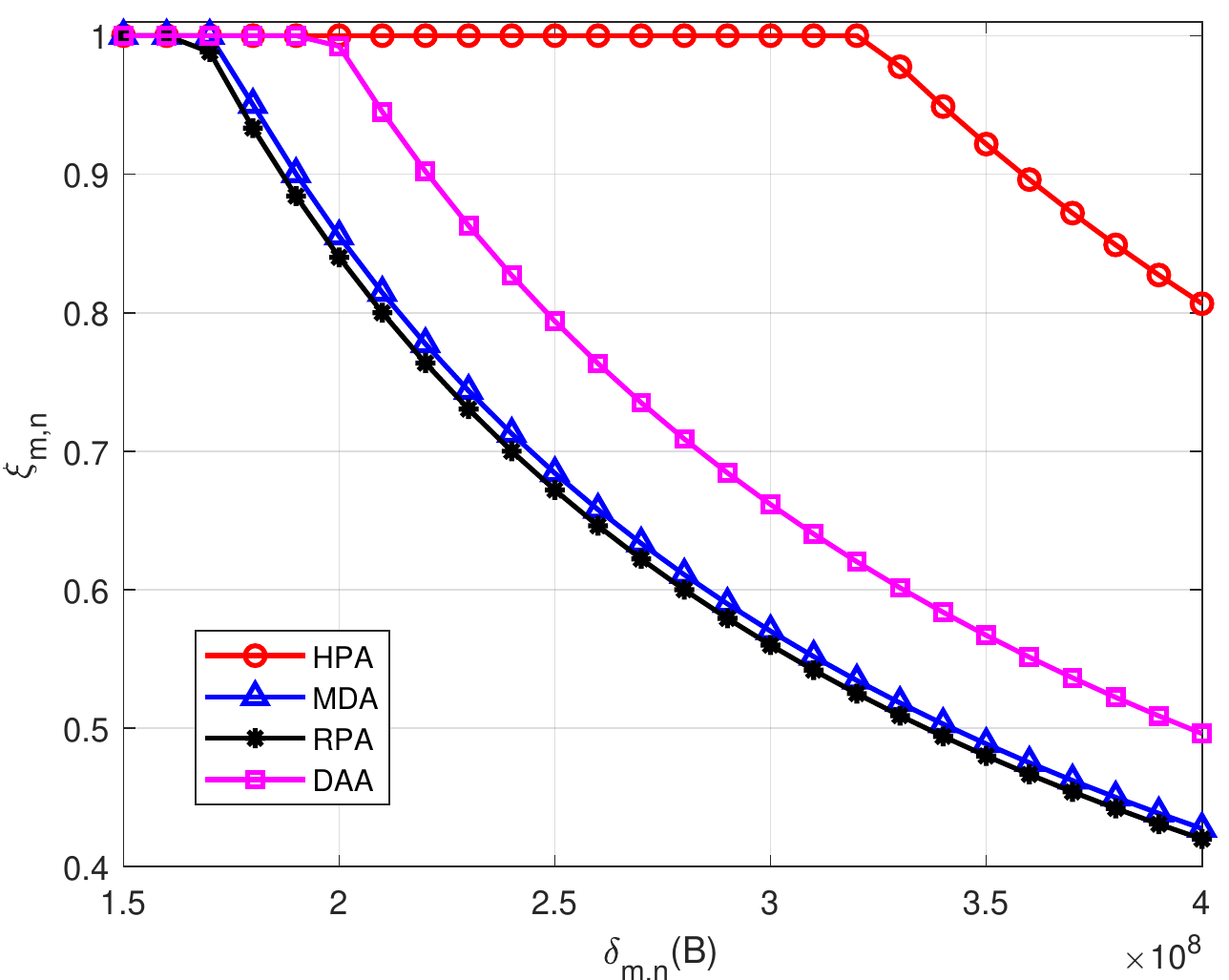}\label{fig:xidelta}}
	\end{tabular}
	\begin{tabular}{@{}c@{}}
			\subfloat[]{
			\includegraphics[scale=.55, trim={4cm 8.5cm  4.5cm  9cm },clip]{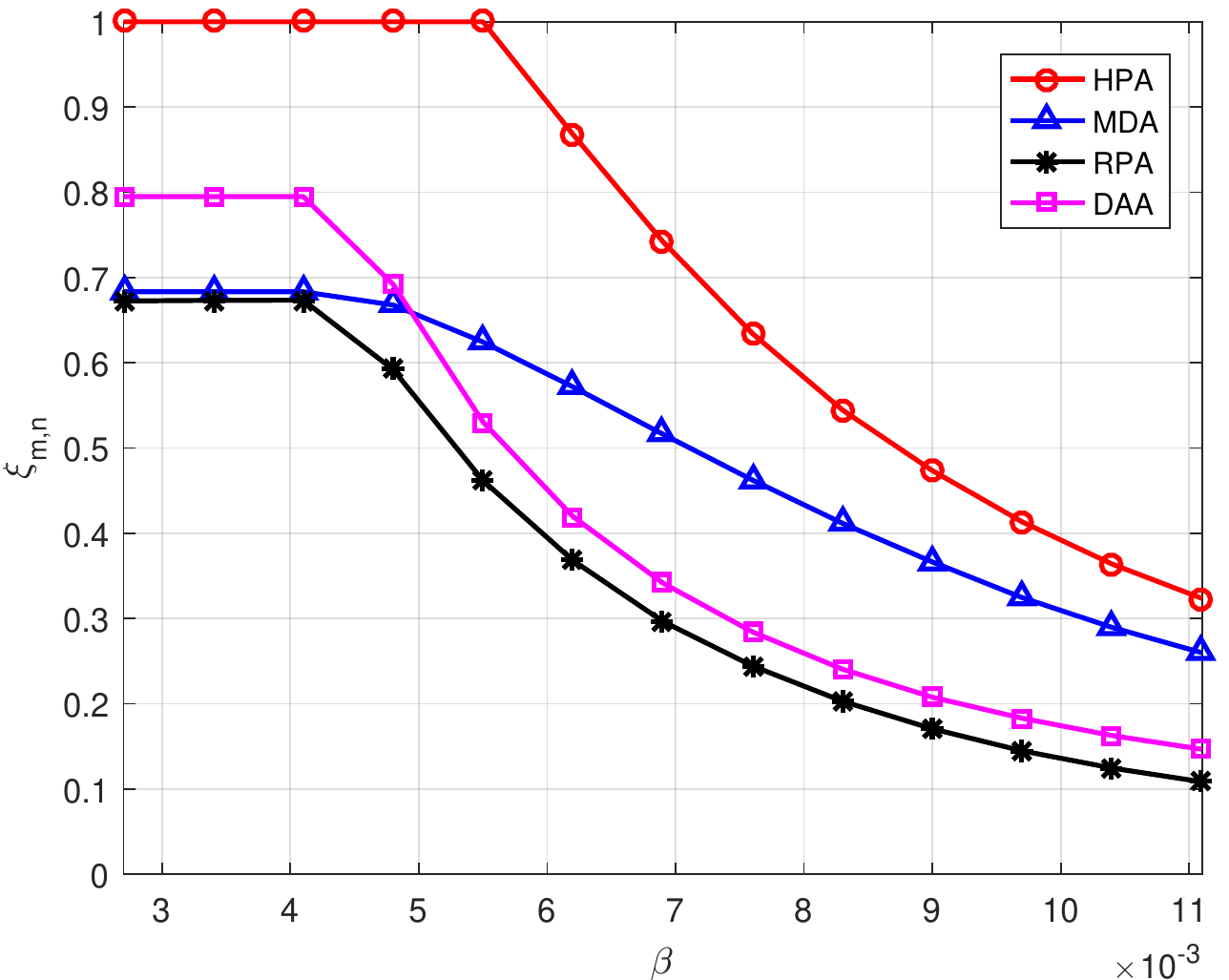}\label{fig:xibeta}}
			\subfloat[]{
			\includegraphics[scale=.55,trim={4.cm 8.5cm  4.5cm  9cm },clip]{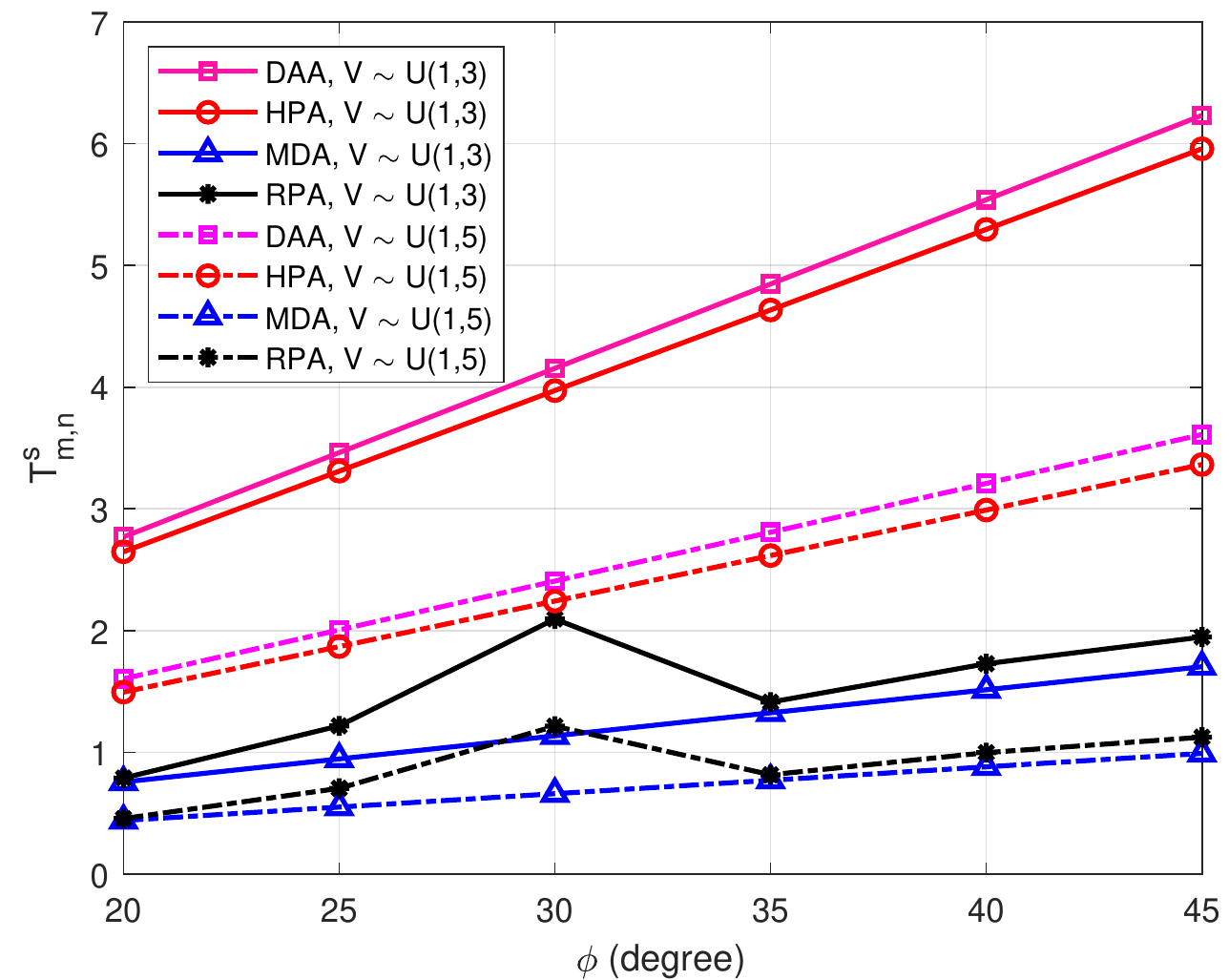}\label{fig:PAlink}}
	\end{tabular}		
	\caption{Performance comparison of the heuristic peer association algorithm with DAA, MDA, RPA: (a) data throughput as a function of data segment availability with antenna beamwidth $\phi=30^{\circ} $, (b) data throughput as a function of data packet size with $p_D = 1$, (c) data throughput as a function of blockage parameter $\beta$, and (d) link stability time of D2D links with $p_D=1$.}
\end{figure*}
\section{Numerical Results}\label{sec:NumericalResults}
In this section, we evaluate the performance of the proposed scheme through extensive simulations. Moreover, to emphasize the importance of designing the proposed scheme in enabling the CCN-based D2D network, its performance is compared with other methods proposed in the literature.
\subsection{Simulation Outline}
\hl{We built our simulator in MATLAB\textsuperscript{\textregistered} consisting of the D2D interaction environment in an area of the size \mbox{$10$ $km$ $\times$ $10$ $km$}, which is \mbox{-–given the transmit power of D2D users-–}  large enough to avoid the boundary effect.
In the simulation environment, D2D users are located uniformly and move according to the random walk model. D2D users' trajectories (speed and direction of movement) are drawn based on i.i.d. uniform random variables. D2D transceivers are equipped with a directional antenna for data transmission in the mmWave band. Also, we assume that all the DTs transmit at the same power. Simulation parameters shown in Table \mbox{\ref{tab:parameters}} are used, unless otherwise specified. }
To thwart the effect of noisy data, simulation results are obtained using the Monte Carlo simulation by averaging over $10,000$ network topologies.
We have outlined evaluation performance using three cases.

\begin{itemize}[noitemsep,topsep=0pt]
\item First we evaluate the performance of the proposed heuristic peer association (HPA) algorithm, in terms of the data throughput, link stability time of D2D links, and the amount of traffic offloaded from the cellular network. Moreover, the performance of the HPA algorithm is compared with minimum-distance peer association (MDA) \cite{xiao2016energy}, where DRs are paired with the closest DT, and random peer association (RPA)\cite{gu2015matching} through which DRs are associated randomly to a DT in its coverage area. \hl{ We also compare HPA's performance with the differed acceptance algorithm (DAA), which is a classic approach to solving the peer association problems that are modeled as a matching game\mbox{\cite{ namvar2015context,perfecto2017millimeter}}.}
 \item Second, the performance of the LLL-based beamwidth selection (LLL-BWS) algorithm in enabling established D2D links to optimize their antenna beamwidth is evaluated. The number of iterations it takes for the algorithm to converge to the optimal solution, and the sum of the network's data throughput are chosen as the performance measures. In addition, the performance of the LLL-BWS is compared with constant beamwidth selection (CBS) \cite{thornburg2016performance}, and random beamwidth selection (RBWS) \cite{ding2017random}. In CBWS, all users implement constant and identical antenna beamwidth, while in RBWS, each user randomly selects a beamwidth.
    \item \hl{Finally, to demonstrate the overall impact of our proposed initialization scheme in improving the network performance, we combine the HPA and LLL-BWS algorithms for peer association and beamwidth selection and compare their performance with other methods used in the literature. }
\end{itemize}

\subsection{Impact of Peer Association Algorithm}
\hl{\mbox{\textbf{Simulation setup I-}} A given DR, known as "test" DR located at the origin \mbox{$(0,0)$} is surrounded by DTs with a density of \mbox{$M = 10$ $km^{-2}$}. The DR's desired data packet is cached in DTs with a probability of \mbox{$p_D$}. The size of the data packet is 300 \mbox{$MB$}. The test DR is paired with a DT in its set of feasible DTs using HPA, DAA, MDA, and RPA algorithms. In case the selected DT in RDA and MDA does not contain the desired data packet, the transmission fails. Note that the probability of LOS link is relatively high due to the small value of blockage parameter $\beta = 0.0027$.}

Figure \ref{fig:PAprob} demonstrates the data throughput of the test DR, as a function of the probability of desired data packet availability with antenna beamwidth of 30$^{\circ}$. It is shown that in general, higher $p_D$ leads to higher data throughput, as the number of successful transmissions is higher. Also, it can be seen that HPA provides significantly higher data throughput compared to DAA, MDA and RPA, due to considering data segment availability, D2D users' velocity and directionality of the links simultaneously. \hl{Although DAA uses the same utility function as HPA, its convergence is time-consuming. The time complexity of DAA is \mbox{$\mathcal{O}(n^2)$} ($n$ is the number of users) to obtain the stable matching for a canonical matching game \mbox{\cite{perfecto2017millimeter}}. However, HPA's time complexity is constant and independent of the number of users, i.e., \mbox{$\mathcal{O}(f)$} ($f$ is the number of trials). Hence,  the time overhead complexity of DAA is significantly high compared to HPA. Consequently, DAA has a lower data rate than HPA, which leads to significantly lower data throughput as per \mbox{\eqref{eq:dataThrou}}. }

\hl{Figure \mbox{\ref{fig:xidelta}} shows the data throughput of the test DR as a function of data packet size with antenna beamwidth of $\phi=30^\circ$ and data packet availability \mbox{$p_D=1$}. It can be seen that as the size of the packet increases the data throughput decreases, however, HPA manages to maintain the data throughput very high and transmit the data packet completely when the data packet size is smaller than \mbox{$\delta_{m,n}<320$ $Mb$}. However, since DAA, RPA and MDA could not verify whether the link is LOS or stable for the duration of data transmission, thus failing to transmit the data packet when the size of the data is greater than 200 \mbox{$MB$}.}

\hl{Figure \mbox{\ref{fig:xibeta}} compares the data throughput of the test DR as a function of blockage parameter \mbox{$\beta$} with antenna beamwidth of \mbox{$\phi=30^\circ$} and data packet availability \mbox{$p_D=1$}. It can be seen that a higher blockage parameter deteriorates the link data throughput. This is in accordance with the fact that higher blockage parameter increases the probability of blocked D2D links and NLOS signal at mmWave band frequencies undergoes severe attenuation \mbox{\cite{bai2014analysis}}. However, since HPA verifies if the link is LoS before establishing the D2D links, it maintains higher data throughput compared to existing approaches.}

Figure \ref{fig:PAlink} shows the link stability time of a given D2D link as a function of antenna beamwidth with different D2D users' velocities. In order to solely analyze the impact of users' velocity on link stability, it is assumed in this scenario that the test DR's desired data packet is cached in all DTs, i.e., $p_D=1$. Figure \ref{fig:PAlink} shows that narrower antenna beams are more prone to misalignment, thus providing less stable D2D links. Also, it can be seen that as D2D users move faster, the D2D link stability time decreases, which is in accordance with \eqref{eq:linkS}. However, the proposed HPA algorithm provides D2D users with higher link stability time compared to MDA and RPA, thus guaranteeing higher data throughput. Although DAA provides a slightly better link stability time due to its semi-exhaustive search nature, DAA's high overhead in matching D2D users, results in significantly lower throughput.
\begin{figure}
\centering
\includegraphics[scale=.53, trim={4cm 8.3cm  4.5cm  8.7cm },clip]{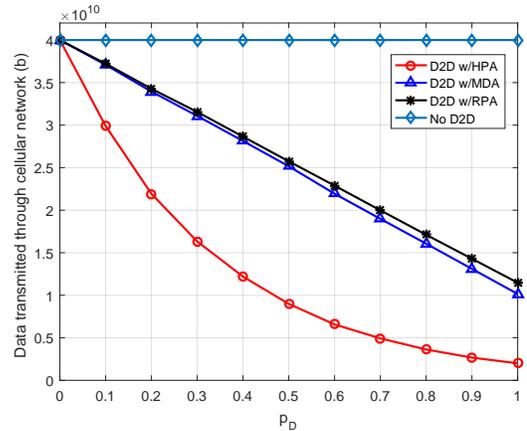}
\caption{Offload data from cellular network with 5 different data packet with size $1$ $Gb$.}
\label{fig:offload}
\end{figure}

\begin{figure}
    \centering
    \vspace{-.3cm}
    \includegraphics[scale=.53, trim={0cm 0cm  0cm  0cm },clip]{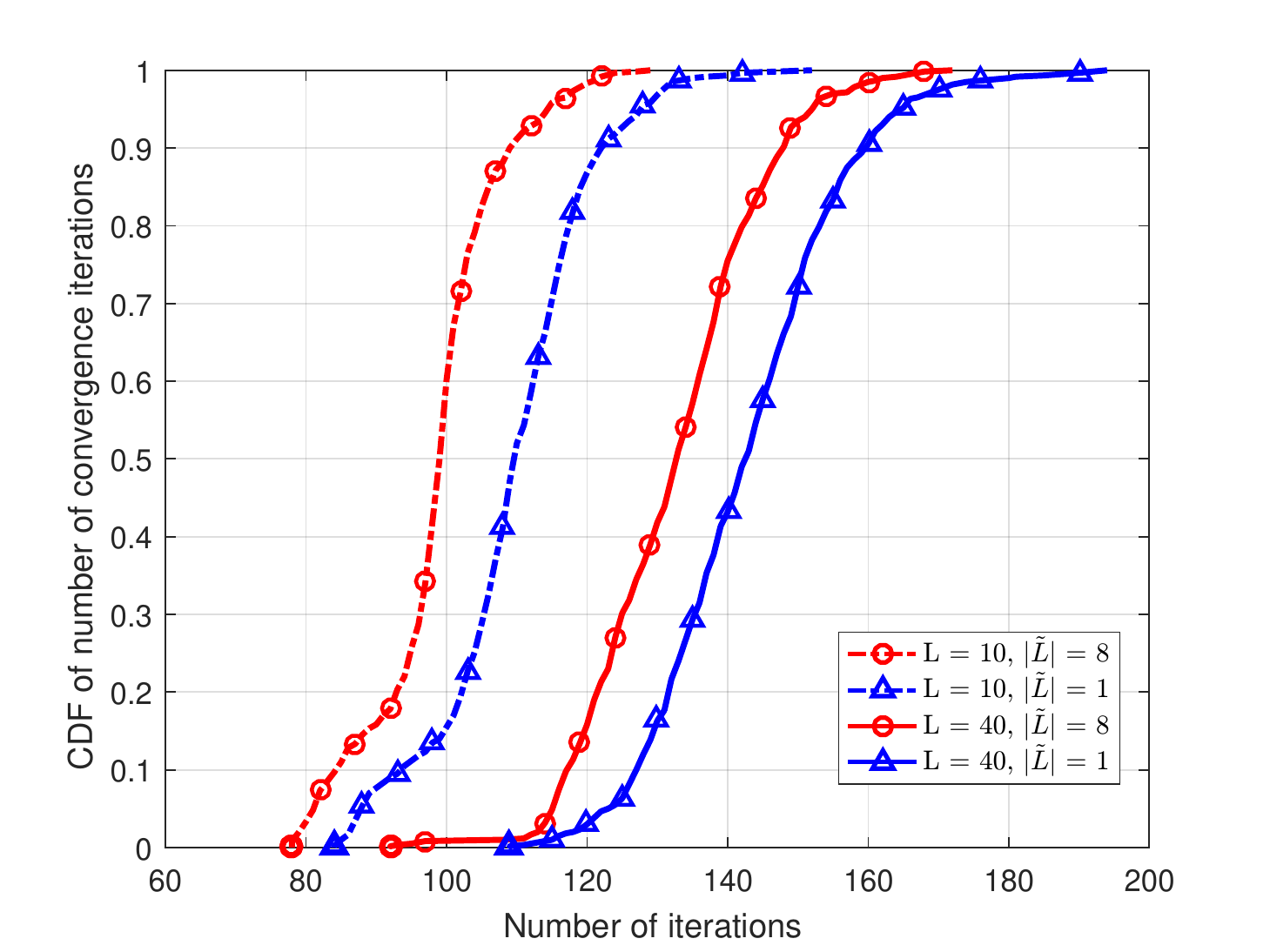}
    \caption{Number of iterations required for the synchronous LLL to converge, for  $L=10$ $km^{-2}$, $L=40$ $km^{-2}$, $|\tilde{\mathcal{L}}|=1$, and $|\tilde{\mathcal{L}}|=8$.}
    \label{fig:cdfOfIterations}
\end{figure}

\hl{\textbf{Simulation setup II-} D2D users are distributed uniformly in the network area with a density of \mbox{$M=N=40$ $km^{-2}$}. D2D users are assigned to be DT or DR with equal probability. DTs are cached randomly with 5 different popular data packets with a size of 1 \mbox{$Gb$}. DRs attempt 3 times to retrieve the desired data packet from a DT in their coverage area before switching to the cellular network, i.e., \mbox{$f=3$}.}

\hl{Figure \mbox{\ref{fig:offload}} shows the amount of data transmitted by the cellular network as a function of antenna beamwidth. It can be seen that utilizing the proposed HPA algorithm significantly offloads the cellular network compared to MDA and RPA algorithms. This is in accordance with the fact that HPA not only guarantees the stability of the D2D link by considering the user-specific context information, but also confirms the availability of the desired data packet prior to establishing a D2D link for data transmission. Therefore, HPA has a higher ability to offload the cellular network compared to its contemporary rivals.}

\subsection{Impact of Beamwidth Selection Algorithm}
\hl{\textbf{Simulation setup III-} In this simulation, $L$ D2D links are established between D2D transceivers that are located uniformly in the network environment. DTs are cached with various data packet sizes. The size of the data packets is distributed uniformly as $\delta_{l}\sim \mathcal{U}(0-300)$ $MB$. DTs are ready to transmit the data packets to their corresponding DRs. Also, corresponding D2D transceivers' wide-level beams are aligned.}

Figure \ref{fig:cdfOfIterations} shows the number of iterations required for the convergence of the synchronous LLL-BWS algorithm with two different D2D link densities, i.e., $L=10$ $km^{-2}$, and $L=40$ $km^{-2}$. It can be seen that a higher number of D2D links requires a higher number of iterations to find the optimal solution. However, this figure shows that increasing the number of simultaneously updating links from $|\mathcal{\tilde{L}}|=1$ to $|\mathcal{\tilde{L}}|=8$ increases the convergence speed of the algorithm significantly. This observation is in agreement with the fact that as more non-interacting D2D links have the opportunity to update their strategies simultaneously, the algorithm converges to the optimal NE faster.
\begin{figure*}
\vspace{-.3cm}
    \centering
    \begin{tabular}{@{}c@{}c@{}}
        \subfloat[]{
            \includegraphics[scale=.43, trim={3cm 7.9cm  3cm  7.4cm },clip]{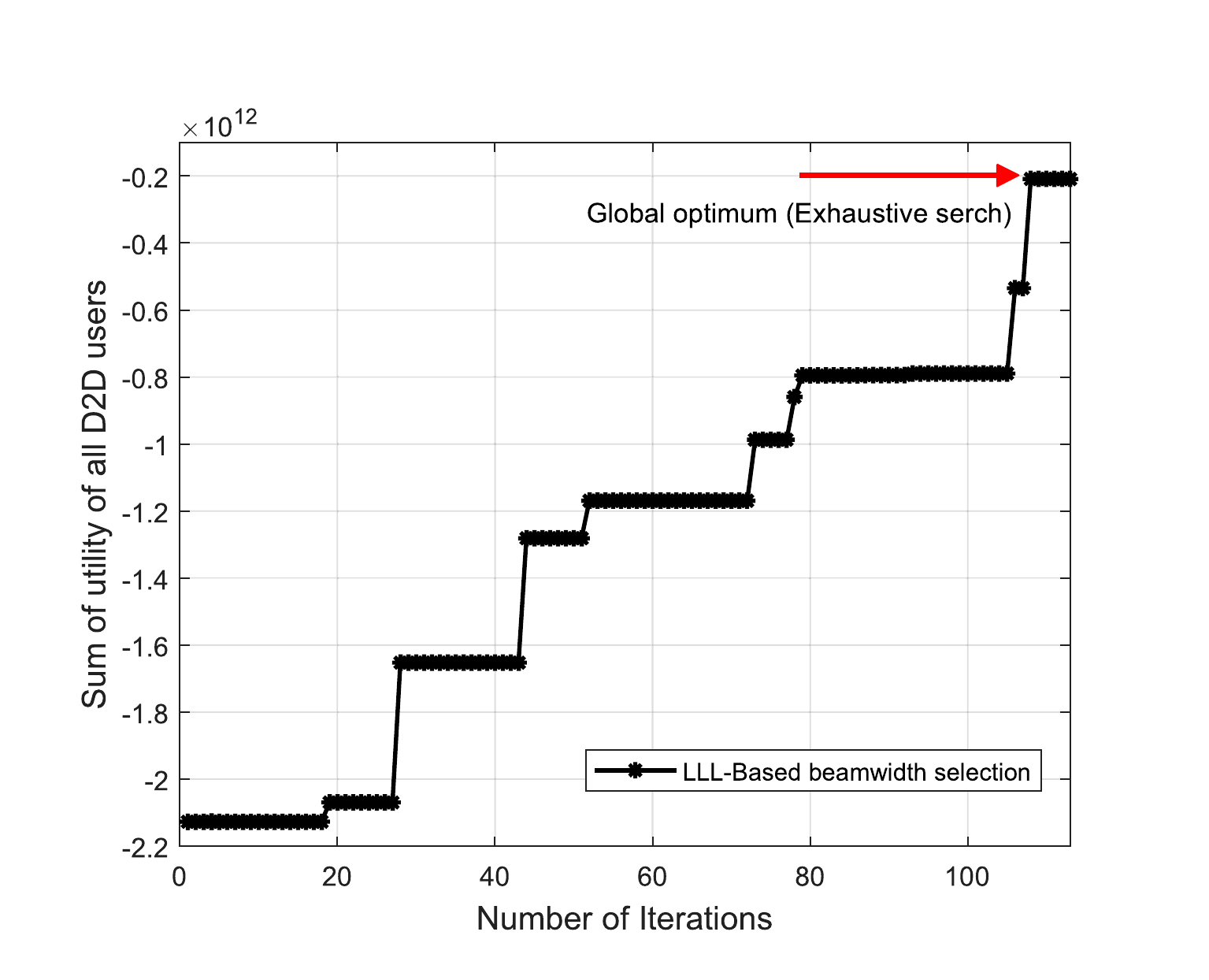}\label{fig:globalUtility}}\hfill
        \subfloat[]{
            \includegraphics[scale=.293, trim={1cm 5cm  1cm  6.35cm },clip]{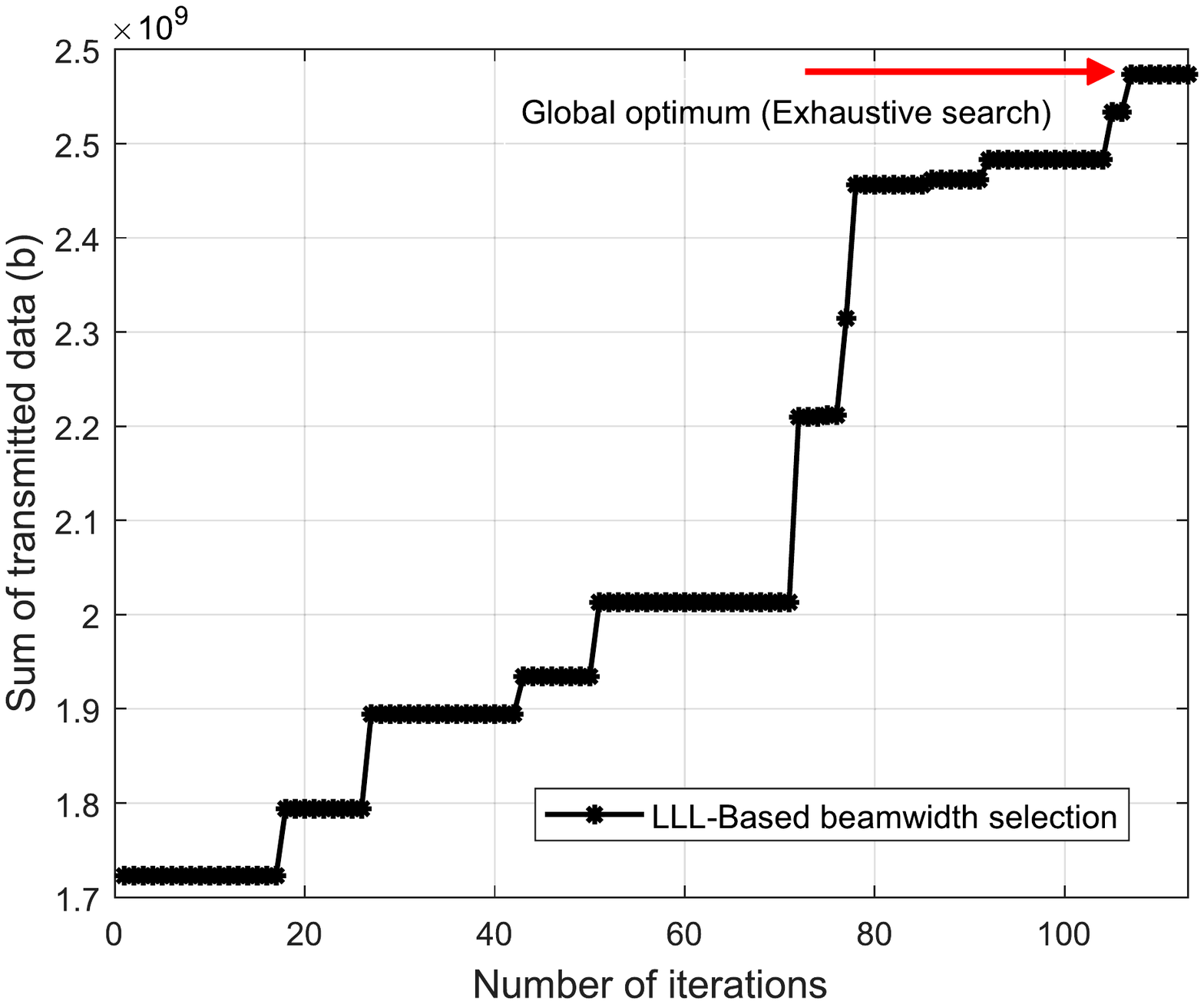}\label{fig:DeliveredData}}\hfill
            \subfloat[]{
            \includegraphics[scale=.47,trim={1cm 0cm  1cm  .7cm },clip]{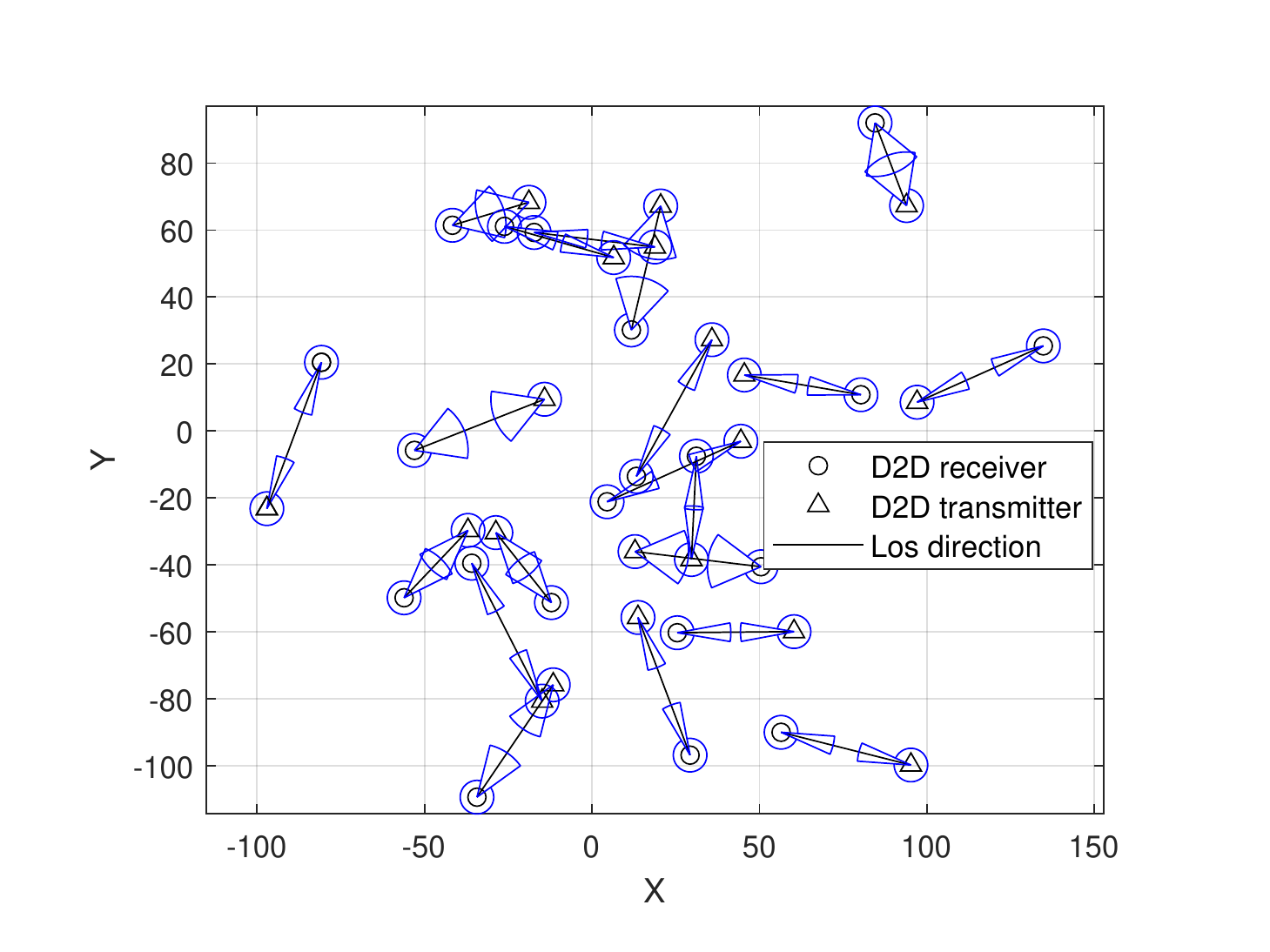}\label{fig:NetSchematic}}%
    \end{tabular}
    \caption{Convergence of the proposed mechanism to the global optimum for D2D link density of $L=20$ $km^{-2}$: (a) sum of network's global utility, (b) sum of network's data throughput, and (c) joint beamwidth selection strategy of D2D links, which is obtained using the proposed synchronous LLL-based beamwidth selection algorithm.}
\label{fig:CDFRateDelay}
\end{figure*}

Figure \ref{fig:globalUtility} and \ref{fig:DeliveredData} show the sum of global utility and sum of delivered data segments for a D2D network with D2D link density $L=20$ $km^{-2}$, respectively. These two figures verify that the utility of the bandwidth selection game $\mathcal{G}_b$ in \eqref{eq:game} is completely aligned with the network optimization problem in \eqref{eq:mainOptprb} with constraints \eqref{eq:antennaCon}-\eqref{eq:mainOptprb_h}. In other words, maximizing utility function in \eqref{eq:totalUt} for all D2D links $l \in \mathcal{L}$ leads to maximizing the network's sum data throughput. In addition, it can be seen that the proposed algorithm converges to the global optimum which is derived using the exhaustive search algorithm. The exhaustive search takes $\mathcal{O}(\prod_{l=1}^L |\Phi_l|)$ to find the optimal solution, while as it can be seen in \ref{fig:globalUtility}. the proposed algorithm converges considerably faster. For example, in this case, with $|\Phi_l|=4$, $4^{20} \sim 10^{12} $ iterations are required to obtain the optimal solution with exhaustive search, while the proposed algorithm found it in only $110$ iterations.
Figure \ref{fig:NetSchematic} depicts a snapshot of the network topology where the D2D links have selected their beamwidth strategy using the proposed LLL-BWS algorithm. It can be seen that D2D users have selected different beamwidths based on the size of the demanded data $\delta_l$, and their geographical position in the network. While without considering these two factors, all D2D users prefer to select the narrowest beamwidth to maximize their antenna gain and thus their achievable data rate.

Figure \ref{fig:comparison} plots the sum of the transmitted data segments for different D2D link density from $L=5$ $km^{-2}$ to $L=40$ $km^{-2}$. The performance of the proposed LLL-BWS algorithm is compared with CBWS and RBWS. For the case of CBWS, the narrowest and the widest beamwidth has been implemented, respectively. It can be seen that the proposed algorithm outperforms CBWS and RBWS algorithms. For example, in the case of $L=30$ $km^{-2}$ the proposed LLL-BWS algorithm improves the performance of the network by 3 times compared to the case that all D2D users adopt their narrowest beamwidth. In addition, it can be seen that when the network becomes congested (here for $L>30$ $km^{-2}$), the performance of the CBWS and RBWS methods deteriorates notably, while the proposed algorithm manages to keep the network performance at a desirable level. In particular, picking the widest beamwidth in the hot spots (congested areas) leads to a significant drop in the amount of transferred data. This is in accordance with the fact that the wider antenna beamwidth increases the interference on nearby users, and eventually decreases the achievable data rate of most of the D2D links.
\begin{figure}[!ht]
    \vspace{-.5cm}
    \centering
    \includegraphics[width = .8 \linewidth,  trim={4.3cm 7.9cm  4.6cm  8.6cm },clip]{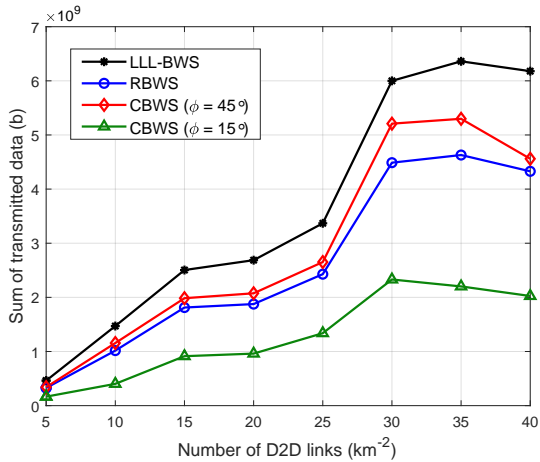}
    \caption{Performance comparison of the proposed beamwidth selection algorithm with constant and random beamwidth selection algorithm, for the network size of $L=5$ $km^{-2}$ to $L=40$ $km^{-2}$.}
    \label{fig:comparison}
\end{figure}

\subsection{Impact of the Proposed Initialization Scheme}
\hl{\textbf{Simulation setup IV-} D2D users are located uniformly in the network area with a density of \mbox{$M=N=40$ $km^{-2}$}. D2D users are assigned to be DT or DR with equal probability. DTs are cached randomly with 5 different popular data packets with a size of 300 \mbox{$MB$}. D2D users implement the proposed initialization scheme. First, D2D users are matched using HPA (Algorithm \mbox{\ref{alg:peerAssoc}}), MDA and RPA. Then, D2D users' antenna beamwidth are selected using the proposed LLL-BWS algorithm (\mbox{Algorithm \ref{alg:beamwidth}}) and CBWS with a beamwidth of $\phi=15^{\circ}$.}

\hl{Figure \mbox{\ref{fig:comparisonMethods}} compares the performance of the proposed initialization scheme against other methods used in the literature. The performance metric is the network's average data throughput which is defined as $\frac{1}{L}\sum_{l=1}^{L}\xi_l$. This figure shows the cumulative distribution function (CDF) of networks' average data throughput. It can be seen that the proposed scheme comprising the HPA and LLL-BWS outperforms the other methods noticeably, since HPA considers the users' trajectory and content availability in user assignment based on \mbox{\eqref{eq:utilityPeer}}. In addition, LLL-BWS enables users to optimize their beamwidth based on the data size and link stability time. The utility function in \mbox{\eqref{eq:totalUt}} is designed in a way that enables D2D users to maximize the network throughput instead maximizing individual's throughput.}

\begin{figure}[!ht]
    \centering
    \includegraphics[width = .9 \linewidth, trim={0cm 0cm  0cm 0cm },clip]{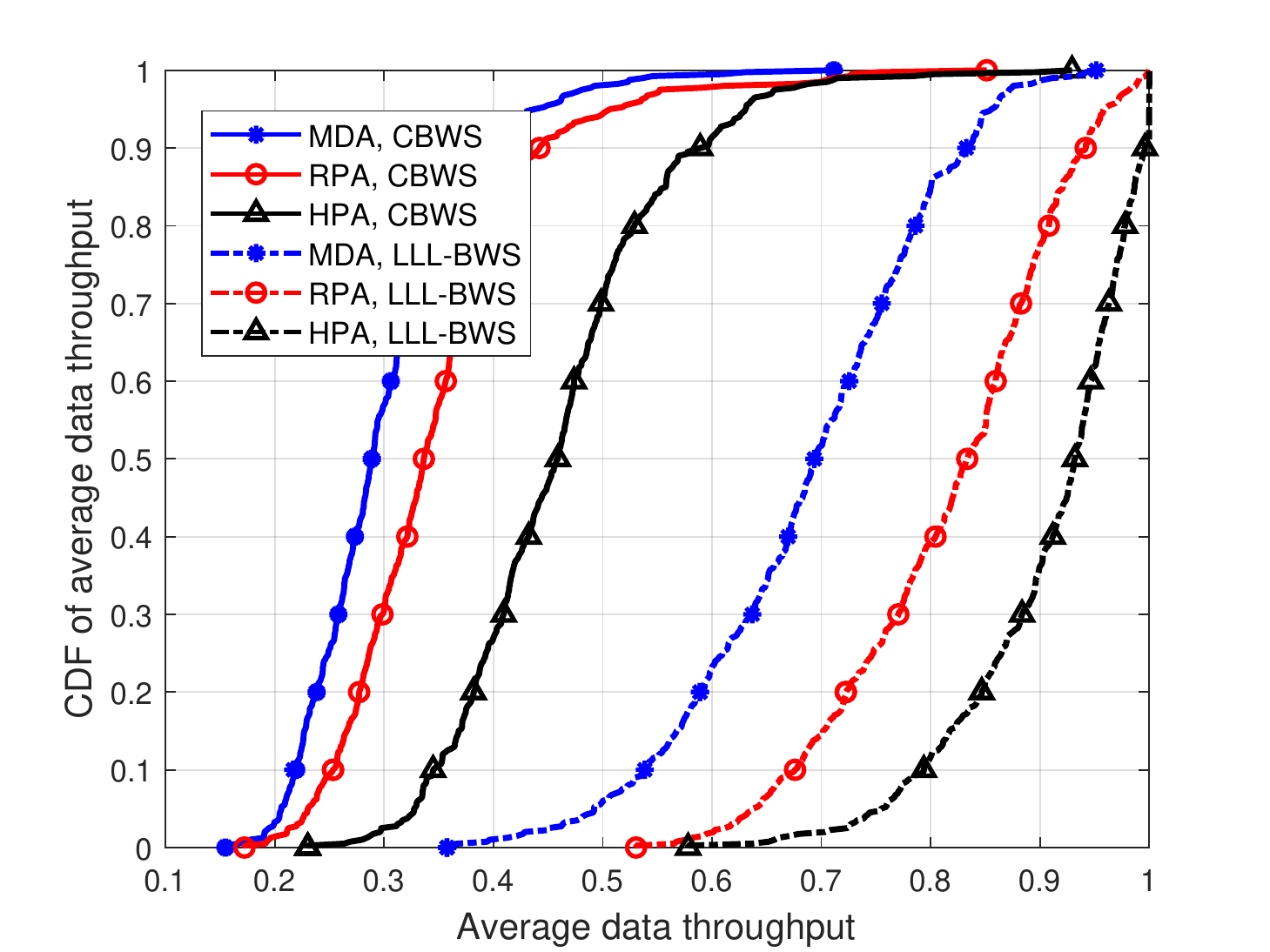}
    \caption{Performance comparison of the proposed initialization scheme: impact of combining heuristic peer association and LLL-based beamwidth selection.}
    \label{fig:comparisonMethods}
\end{figure}
\vspace{-5pt}
\section{Conclusions and Future Work}\label{sec:conclusion}
A novel decentralized scheme is proposed to enable D2D users to perform the initialization process in a CCN-based mmWave D2D network. The proposed scheme comprises of two phases, namely, heuristic peer association algorithm and synchronous beamwidth selection algorithm. The context-aware peer association algorithm is low-overhead with a low computational load and enables peer association in a decentralized manner. Following the peer association, antenna beamwidth optimization is performed considering the trade-off between antenna beamwidth and data throughput in directional communication.
A synchronous LLL-based algorithm is proposed to obtain the joint beamwidth selection strategy of all users to maximize the network data throughput.
the performance of the proposed scheme is evaluated through extensive Monte Carlo simulations. Simulation results show that the proposed initialization scheme significantly improves the network performance compared to other methods in the literature.

\hl{Future research includes assessing the performance of the initialization scheme in real scenarios, in particular indoor applications. In such scenarios, D2D links are more susceptible to misalignment and have shorter link stability time. The proposed beamwidth selection algorithm in this work will manage to compensate for the stability time by selecting the proper antenna beamwidth and render higher performance gains.}

\section*{Acknowledgments}
The authors would like to acknowledge the support from
Air Force Research Laboratory and OSD for sponsoring this
research under agreement number FA8750-15-2-0116.

\ifCLASSOPTIONcaptionsoff
  \newpage
\fi

\bibliographystyle{IEEEtran}
\bibliography{main.bib}

\begin{IEEEbiography}[{\includegraphics[width=1in,height=1.25in,clip,keepaspectratio]{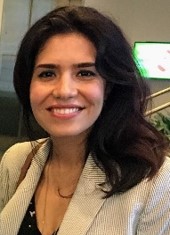}}]{Niloofar Bahadori}
received the B.Sc. degree in electrical and electronics engineering from Isfahan University, in 2011, the M.Sc. degree (with Hons.) in electrical and radio frequency (RF) engineering from Semnan University in 2013.  She is currently pursuing a Ph.D.
degree at North Carolina A\&T State University, Greensboro, NC, USA.  Her current research interests include device-to-device (D2D) and machine-to-machine (M2M) communication, mmWave band communication the Internet of Things (IoT), the applications of machine learning in improving wireless networks, and game theory. She is the recipient of
the 2019 IEEE Wireless Telecommunications Symposium (WTS) Best Paper
Award.
\end{IEEEbiography}

\begin{IEEEbiography} [{\includegraphics[width=1in,height=1.25in,clip,keepaspectratio]{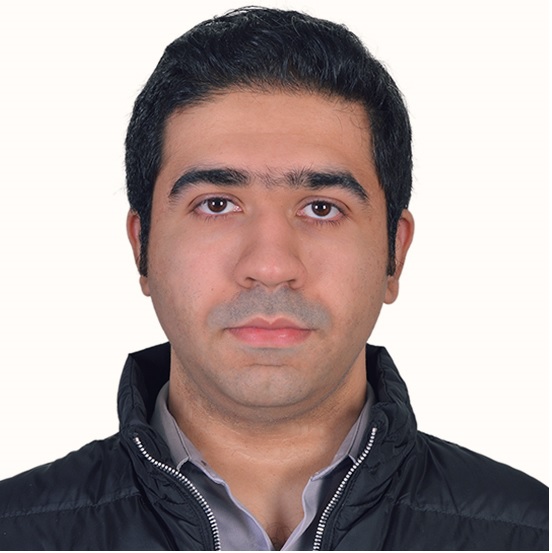}}]{Mahmoud Nabil}
 is an Assistant Professor in the Department of Electrical and Computer Engineering, North Carolina A and T University. He received his Ph.D. degree in Electrical and Computer Engineering from Tennessee Tech University, Cookeville, Tennessee in August 2019. He received his B.S. degree and the M.S. degree with honors in Computer Engineering from Cairo University, Egypt in 2012 and 2016, respectively. He published many journals and conferences in different prestigious venues such as IEEE internet of things journal, IEEE transactions of dependable and secure computing, IEEE Access, international conference on communication (ICC), international conference on pattern recognition (ICPR), and international conference on wireless communication (WCNC). His research interests include security and privacy in smart grid, machine learning applications, vehicular Ad Hoc networks, and blockchain applications.
\end{IEEEbiography}

\begin{IEEEbiography} [{\includegraphics[width=1in,height=1.25in,clip,keepaspectratio]{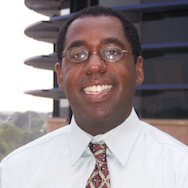}}]{Brian Kelley}
 received his BSEE from Cornell University's College of Electrical Engineering, Ithaca NY, and his MSEE and PhDEE from Georgia Tech in 1992. He spent 10 years with Motorola R\&D as principal architect of Wi-Fi, LTE, cellular platforms and as representative to the 3GPP-RAN Standards. Since 2007, he has been an Associate Professor of Electrical and Computer Engineering at the University of Texas at San Antonio (UTSA). He was Associate Editor of the IEEE System Journal from 2011-2013, Technical Program Committee Chair for IEEE GLOBECOM in 2015, Sabbatical Employee of DoD in Washington D.C. from 2015-2016 and Summer Faculty Fellow at ORNL in the Quantum Information System in 2015 and 2017. He has received over \$2.6M in research funding from ONR, has numerous IEEE publications and 11 US patents, and his current research focuses on 5G Communications and Physical Layer Security. He is also a member of the 5G steering group for JBSA’s Electromagnetic Defense Initiative, serves as UTSA’s representative to the National Spectrum Consortium and is a member of Tau Beta Pi and Eta Kappa Nu.
\end{IEEEbiography}

\begin{IEEEbiography} [{\includegraphics[width=1in,height=1.25in,clip,keepaspectratio]{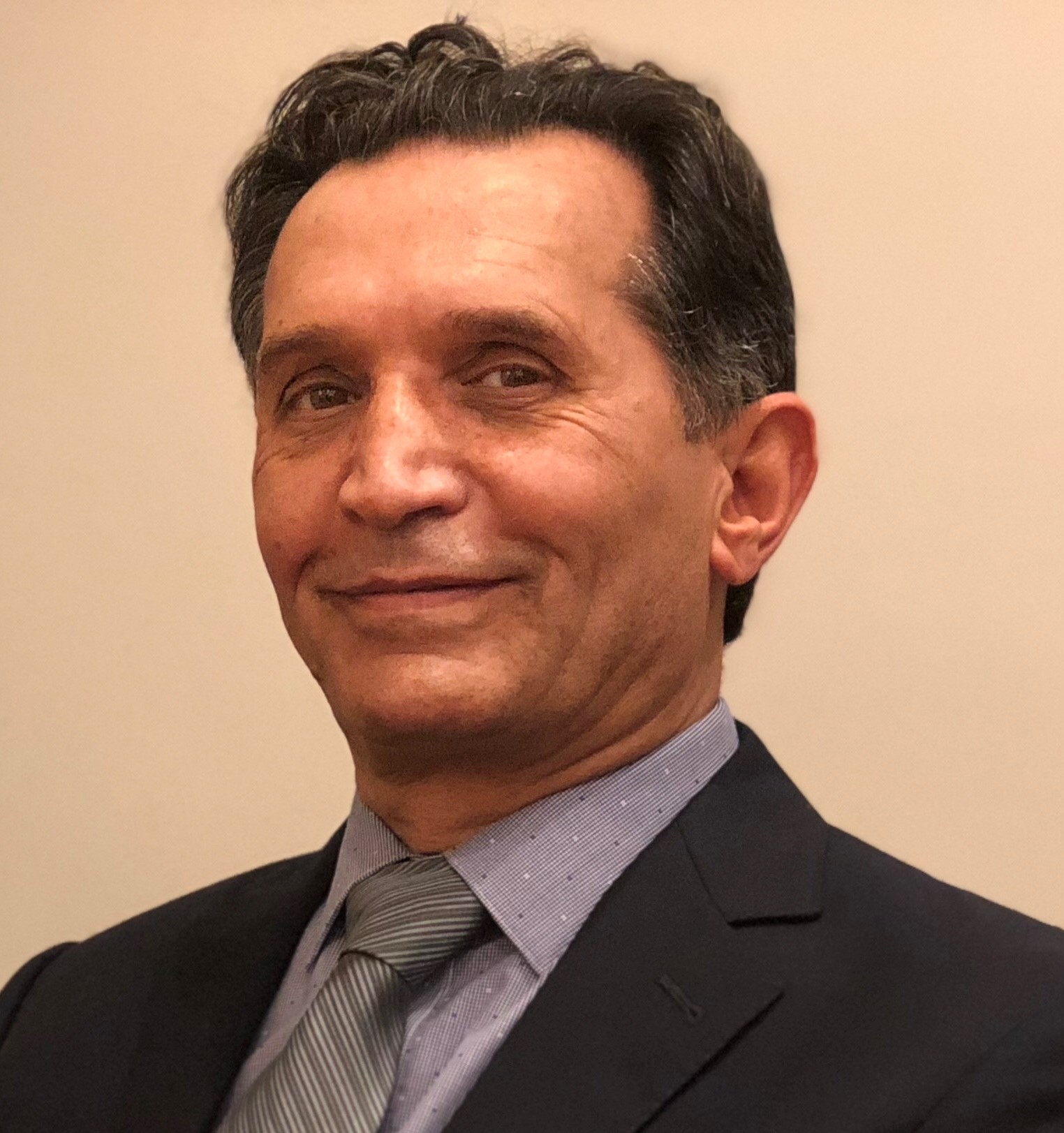}}]{Abdollah Homaifar}
 received the B.S.
and M.S. degrees from the State University of
New York at Stony Brook, in 1979 and 1980,
respectively, and the Ph.D. degree from the
University of Alabama, in 1987, all in electrical engineering. He is currently the NASA
Langley Distinguished Chair Professor and
the Duke Energy Eminent Professor with the
Department of Electrical and Computer Engineering, North Carolina A\&T State University (NCA\&TSU). He is also the Director of the Autonomous Control and Information Technology Institute and the Testing, Evaluation, and Control of Heterogeneous Large-Scale Systems of Autonomous Vehicles (TECHLAV), NCA\&TSU. His current research interests include machine learning, unmanned aerial vehicles (UAVs), testing and evaluation of autonomous vehicles, optimization, and signal processing. Through his research, he has obtained funding in excess of 30 million from various U.S. funding
agencies. He has written more than 350 technical publications including book chapters and journal and conference papers. He is a member of the IEEE Control Society, Sigma Xi, Tau Beta Pi, and Eta Kapa Nu. He also serves as an Associate Editor for the Journal of Intelligent Automation and
Soft Computing. He serves as a Reviewer for the IEEE Transactions on Fuzzy Systems, Man Machines and Cybernetics, and Neural Networks.
\end{IEEEbiography}

\end{document}